\documentclass[journal]{new-aiaa}
\usepackage[utf8]{inputenc}

\usepackage{graphicx}
\usepackage{epstopdf}
\usepackage{float}

\usepackage{amsmath,amsthm,bm,bbm,latexsym,subcaption,hyperref}
\usepackage[version=4]{mhchem}
\usepackage[english]{babel}
\usepackage[title]{appendix}
\usepackage[pagewise]{lineno}
\usepackage{siunitx}
\usepackage{longtable,tabularx}
\usepackage[ruled,linesnumbered]{algorithm2e}

\usepackage{etoolbox}% for \forcsvlist

\usepackage[dvipsnames]{xcolor}
\graphicspath{{./figures/}} % LM

\DeclareGraphicsExtensions{.pdf,.eps}

\newcommand{\dfn}{\triangleq}

\DeclareMathOperator{\sign}{sign}
\DeclareMathOperator{\sat}{sat}

\newcommand{\meas}{\ensuremath{\mathcal{Z}}}

\let\abs=\envert

\newtheorem{theorem}{Theorem}[section]

\newtheorem{lemma}[theorem]{Lemma}

\newtheorem{assumption}{Assumption}
\newtheorem{definition}{Definition}

\newtheorem{remark}{Remark}

\newcommand{\ProbM}[2]{P_M[#1,#2]}

\newcommand{\DeclareMyOperator}[1]{%
  \expandafter\DeclareMathOperator\csname #1\endcsname{#1}
}
\forcsvlist{\DeclareMyOperator}{%
   supp,
}

\newcommand{\Mcr}{\mathbb{M}^c_r}

\newif\iftracked
\newif\ifsupp
\newif\ifinclude
\trackedfalse \supptrue \includefalse   % (iii) arXiv version (clean text, with supplementary appendices)
\iftracked
  \newcommand{\rev}[1]{\textcolor{blue}{#1}}
\else
  \newcommand{\rev}[1]{#1}
\fi
\title{Comprehensive Approach to Directly Addressing Estimation
  Delays in Stochastic Guidance}

\author{Liraz Mudrik\footnote{\rev{Visiting Lecturer, Stephen B. Klein Faculty of Aerospace Engineering. This work is part of Dr. Mudrik’s Doctoral Research at the Technion; liraz109@gmail.com. Member AIAA (Corresponding Author).}}}

\author{Yaakov Oshman\footnote{Professor Emeritus, Stephen B.\ Klein
    Faculty of Aerospace Engineering;
    \texttt{yaakov.oshman@technion.ac.il}. Fellow AIAA.}}
\affil{Technion--Israel Institute of Technology, 3200003 Haifa,
  Israel}

\begin{document}

\iftracked
    \linenumbers
\fi

\maketitle

\ifinclude\else
    \renewcommand{\include}{\input}
\fi

\begin{abstract}
In realistic pursuit–evasion scenarios, abrupt target maneuvers generate unavoidable periods of elevated uncertainty that result in estimation delays.  
Such delays can degrade interception performance to the point of causing a miss. 
Existing delayed-information guidance laws fail to provide a complete remedy, as they typically assume constant and known delays.  
Moreover, in practice, they are fed by filtered estimates, contrary to these laws' foundational assumptions.  
We present an overarching strategy for tracking and interception that explicitly accounts for time-varying estimation delays.  
We first devise a guidance law that incorporates two time-varying delays, thereby generalizing prior deterministic formulations.  
This law is driven by a particle-based fixed-lag smoother that provides it with appropriately delayed state estimates.  
Furthermore, using semi-Markov modeling of the target’s maneuvers, the delays are estimated in real time, enabling adaptive adjustment of the guidance inputs during engagement. 
The resulting framework consistently conjoins estimation, delay modeling, and guidance.  
Its effectiveness and superior robustness over existing delayed-information guidance laws are demonstrated via an extensive Monte Carlo study.
\end{abstract}

\section{Introduction}
% PI GL - DGL1
Pursuit–evasion differential games have long served as a cornerstone
for developing guidance laws.  In linearized deterministic settings
with bounded controls, these games yield optimal strategies known as
the differential game guidance laws (DGLs)~\cite{gutman_optimal_1979}, e.g.,
DGL0 and DGL1.  In particular, the DGL1 law guarantees hit-to-kill
performance when the pursuer enjoys sufficient maneuverability and
agility advantages~\cite{shinar_solution_1981}.  However, these
results fundamentally rely on the assumption of perfect and
instantaneous state information for both players.  Realistic
interception scenarios violate this assumption, as sensing is noisy,
incomplete, and inherently delayed, rendering the deterministic
performance guarantees unattainable in practice.

% Estimation Process - Time Delay
In realistic stochastic settings, noisy and incomplete sensor
measurements force the pursuer to infer the evader's state through an
estimator.  When a highly maneuverable evader executes an abrupt
change in its acceleration command, the estimator exhibits an
inevitable detection delay before resolving the new maneuver.  This
delay creates an \emph{uncertainty interval} following each abrupt
maneuver, during which the estimated state cannot yet reflect the
evader's true behavior.  Within this interval, we refer to the
estimate as \emph{delayed}, implying that it is associated with a past
state, preceding the target's evasion maneuver.  A well-timed
bang-bang evasion maneuver can, therefore, induce substantial miss
distance, even against advanced guidance laws such as
DGL1~\cite{shaferman_stochastic_2016}, and a sophisticated evader can
even deliberately exploit this
effect~\cite{shaferman_near-optimal_2021}.  Counteracting the adverse
effects of the estimation delay obviously requires estimating it in
real time.  A crude, offline estimate of the estimation delay was
developed, based on first principles, in~\cite{hexner_temporal_2008}
for linearized dynamics with additive Gaussian noise; however, to the
best of the authors' knowledge, no method currently exists in the
guidance literature for estimating this delay in real time from actual
measurement data acquired during the engagement.

% Delayed PI GL - DGLC 
In the absence of online knowledge of the estimation delay, several
delayed-information guidance laws have been developed that treat the
delay as a known, fixed tuning parameter.  Shinar and
Glizer~\cite{shinar_solution_1999} derived an analytical solution for
a pursuit--evasion game with delayed information, which was
subsequently used in~\cite{shinar_nonorthodox_2002} to construct a delay-compensated variant of DGL1, called DGLC, that improves worst-case
performance relative to both DGL0 and DGL1 in stochastic settings.
However, DGLC exhibits significant conceptual shortcomings.  First, it
assumes that the delay is time-invariant, despite the fact that
estimator delays are inherently
time-varying~\cite{hexner_temporal_2008}. Second, it treats the delay
purely as a design parameter, rather than estimating it adaptively
from the available data. Third, it models only the evader's
acceleration as delayed, even though, in practice, the delay affects
additional states.  Moreover, the law is typically driven by the
estimator's current (filtered) output, which is not aligned with the
delayed information assumed in its derivation.  An extension
incorporating time-varying delays was proposed
in~\cite{shinar_linear_2002}, but it addresses only the first of the aforementioned issues
and, therefore, yields only limited improvement.
Glizer and Turetsky~\cite{glizer_linear_2009} demonstrated that the
estimated relative velocity perpendicular to the line of sight is also
subject to delay.  This effect is not accounted for in the DGLC law
and contributes to its degraded performance in stochastic scenarios.
By solving a linear differential game with bounded controls and two
information delays, \cite{glizer_linear_2009} derived a doubly-compensated guidance law, called DGLCC, which achieves superior worst-case performance relative
to DGLC.  However, although it incorporates the additional velocity
delay, the DGLCC law still suffers from the remaining limitations of
DGLC: it treats the delays as known constants, and it lacks a
mechanism to estimate or adapt them based on real-time measurement
data.

% Smoothing
Treating the estimation delays as constants throughout the
interception effectively means that the estimate is regarded by the
guidance law as uniformly delayed throughout the engagement.  
In reality, the estimate is delayed only during the uncertainty interval
that follows an abrupt evasion maneuver, which ends once the estimator
detects the maneuver.  Driving the guidance law with the estimator's
current-time (filtered) outputs throughout the engagement, as done
in~\cite{shinar_nonorthodox_2002,glizer_linear_2009}, then
amounts to 
% A major shortcoming of the implementation of
% existing delayed-information guidance laws is that whereas they
% require past state information, consistent with the delays assumed in
% their derivation, they are, neverthless, typically driven by the
% estimator's current (filtered)
% output~\cite{shinar_nonorthodox_2002,glizer_linear_2009}. This
% effectively treats the estimate as though it were uniformly delayed
% throughout the engagement.  In reality, the estimate is delayed only
% during the uncertainty interval that follows an abrupt evasion
% maneuver, which ends once the estimator detects the maneuver.  Using
% current-time estimates outside this interval therefore introduces
introducing incorrectly timed information into the guidance loop,
which may degrade the overall interception performance.  Clearly, a
new approach is needed that, from the guidance law's perspective,
treats the estimation delays as time-varying during the engagement,
and, from the estimator's perspective, feeds the guidance law with
correctly timed estimates.
We present an integrated solution to the aforementioned problems that
comprises three main elements. First, we derive a new guidance law
that can handle \emph{time-varying} but known estimation delays in
both the evader's acceleration and the relative velocity. Effectively
generalizing the DGLCC guidance law, we do this by solving a
differential game with bounded controls and two time-varying delays. 
Second, we introduce a novel method for estimating, in real time, the
time-varying estimation delays, which constitute key inputs required
by the aforementioned delayed-information guidance law.  Inspired
by~\cite{blom_hybrid_1991}, our delay estimator is based on modeling
the maneuver-switching mechanism as a semi-Markov process and
augmenting the estimator with a sojourn-time state. This provides
real-time, measurement-based estimation of the uncertainty interval
following each target maneuver, and, in turn, enables real-time
estimation of the time-varying information delays.
Finally, we employ the computationally efficient fixed-lag particle
smoother of~\cite{kitagawa_monte_1996} to provide delay-consistent
state estimates using all measurements within the estimated
uncertainty interval. These smoothed estimates are then used to drive
the two-delay guidance law, forming a holistic framework that unifies
guidance, delay estimation, and appropriate state smoothing in a
structurally consistent manner.

% Structure
The remainder of this paper is organized as follows: The problem is
mathematically formulated in Sec.~\ref{sec:Prob_Form}.
Section~\ref{sec:GL} then develops the perfect-information guidance
law with two time-varying delays.
% Section~\ref{sec:integrated_arch} presents the integrated estimation and guidance architecture, detailing the tracking interface, delay-consistent smoothing, and the robust design optimization strategy.
Section~\ref{sec:est} introduces the particle-based estimation and
smoothing framework, including the real-time uncertainty interval
estimator and the resulting estimation delay estimator.
Section~\ref{sec:design} discusses design considerations and provides
an overview of the complete estimation-guidance scheme.
Section~\ref{sec:sim} presents a comparative Monte Carlo study
evaluating the performance of DGL1, DGLC, and the newly proposed
unified guidance scheme.  Concluding remarks are given in
Sec.~\ref{sec:concl}.  %For completeness, the DGL1 guidance law, with supporting derivations, is summarized in Appendix~\ref{sec:DGL1}.

\section{Problem Formulation}
\label{sec:Prob_Form}
This section mathematically formulates the problem, including the
nonlinear kinematical and dynamic model, the measurement model, and
the linearized model used, in what follows, to solve the
delayed-information differential game.

\subsection{Nonlinear Kinematics and Dynamics}
\label{sec:nonlinear_dyn}
A single pursuer--single evader interception scenario is considered. Figure~\ref{fig:Planar-engagement-geometry} shows a schematic view of the geometry of the assumed planar endgame scenario, where $X_{I}$-$O_{I}$-$Y_{I}$ is a Cartesian inertial reference frame. The pursuer and the evader are denoted by $P$ and $E$, respectively. Variables associated with the pursuer and the evader are denoted by additional subscripts $P$ and $E$, respectively. The speed, lateral acceleration, and path angle are denoted by $V$, $a$, and $\gamma$, respectively.  The slant range between the pursuer and the evader is $\rho$, and the line-of-sight's (LOS) angle, measured from the $X_{I}$ axis, is $\lambda$.

\begin{figure}[tbh]
  \centering
  \includegraphics[width=3.25in]{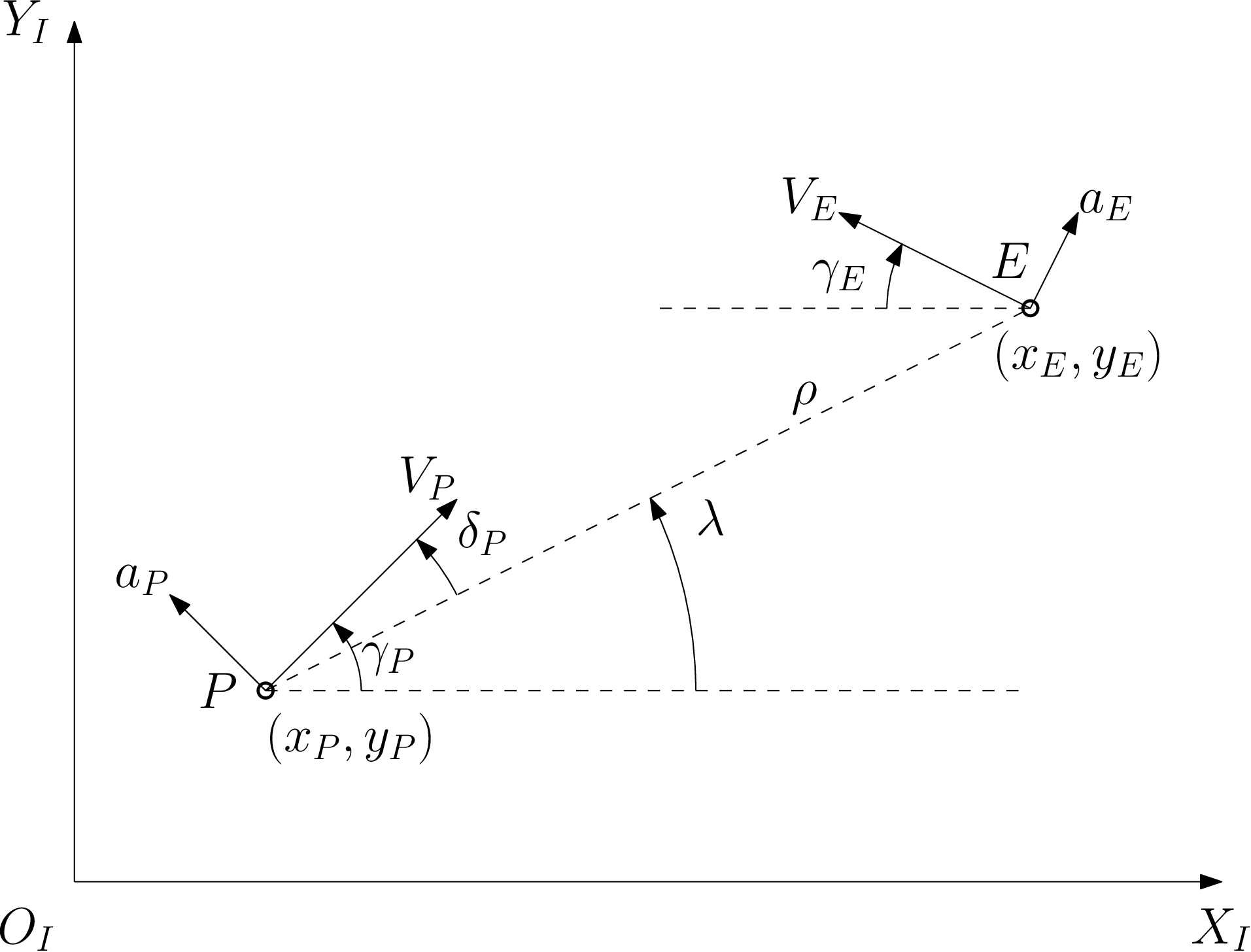}
  \caption{Planar engagement geometry.}
  \label{fig:Planar-engagement-geometry}
\end{figure}

We use the following, commonly accepted, underlying assumptions:
\begin{assumption}
    The pursuer and the evader can be represented as point masses.
\end{assumption}
\begin{assumption}
    \label{ass2} The pursuer's own path angle and lateral acceleration are known (e.g., via the pursuer's own navigation system).
\end{assumption}
\begin{assumption}
    The pursuer's and the evader's speeds, $V_{P}$ and $V_{E}$, respectively, are known and time-invariant.    
\end{assumption}
\begin{assumption}
    \label{ass4} 
    The pursuer and the evader possess first-order dynamics with known time constants, $\tau_{P}$ and $\tau_{E}$, respectively.
\end{assumption}
\begin{assumption}
    \label{ass5} 
    The lateral acceleration bounds of the pursuer and the evader are known constants, $a_{P}^{\max}$ and $a_{E}^{\max}$, respectively, i.e., the acceleration commands satisfy
	\begin{subequations}
		\begin{align}
			a^{c}_{P}(t) &=  a^{\max}_{P}u(t),\qquad \abs{u(t)} \leq 1  \\
			a^{c}_{E}(t) &=  a^{\max}_{E}v(t),\qquad \abs{v(t)} \leq 1 
		\end{align}
	\end{subequations}
	where $u$ and $v$ are the controls of the pursuer and the evader, respectively.
\end{assumption}
\begin{assumption}
    \label{ass6} 
    The evader sets its acceleration command,   $a^{c}_{E}$, from a known \rev{admissible} set \rev{comprising
          $M$ possible maneuvers.  This finite-mode assumption is
          consistent with the multiple-model estimation framework
          employed in the sequel, in which each maneuver corresponds
          to a discrete mode of the estimator.  Consider, for
            example, the case of $M=2$. In this case, the target
            performs bang-bang maneuvers, known to be the optimal,
            perfect-information evasion strategy in
            bounded-acceleration pursuit--evasion games, and its associated
            acceleration command is
    \begin{equation}
    	a^{c}_{E}=\begin{cases}
    		+a_{E}^{\max} & m=1\\
    		-a_{E}^{\max} & m=2
    	\end{cases}
    	\label{Switch_uT}
    \end{equation}
    where $m$ denotes the target maneuver mode.
    }
\end{assumption}
Following common practice, we use polar coordinates to define the equations of motion (EOM). The pursuer's state vector used for the estimation is
\begin{equation}
	\textbf{x}_{P} = \begin{bmatrix}
		\rho & \lambda & \gamma_{E} & a_{E}
	\end{bmatrix}^{T}
	\label{eq:StateVec1}
\end{equation}
and the corresponding EOM are
\begin{subequations}
	\begin{align}
		\dot{\rho}       & = V_{\rho}                                       \\
		\dot{\lambda}    & = \frac{V_{\lambda}}{\rho}                       \\
		\dot{\gamma}_{E} & = \frac{a_{E}}{V_{E}}                            \\
		\dot{a}_{E}      & = -\frac{a_{E}}{\tau_{E}}+\frac{a^{c}_{E}}{\tau_{E}}
	\end{align}
	\label{eq:2}
\end{subequations}
where
\begin{subequations}
	\begin{align}
		V_{\rho}    & =  - (V_{P}\cos\delta_{P}+V_{E}\cos\delta_{E}) \\
		V_{\lambda} & = -V_{P}\sin\delta_{P}+V_{E}\sin\delta_{E}     \\
		\delta_{P}  & = \gamma_{P}-\lambda, \qquad  \delta_{E} = \gamma_{E} + \lambda.
	\end{align}
	\label{eq:EOM1}
\end{subequations}
Based on the underlying assumptions, the pursuer's path angle and
lateral acceleration satisfy the following evolution equations:
\begin{subequations}
	\begin{align}
		\dot{\gamma}_{P}= & \frac{a_{P}}{V_{P}}                            \\
		\dot{a}_{P}=      & -\frac{a_{P}}{\tau_{P}}+\frac{a^{c}_{P}}{\tau_{P}}.
	\end{align}
	\label{eq:EOM2}
\end{subequations}
%where $a^{c}_{P}$ is the pursuer's acceleration command.

\subsection{Measurement Model}

We assume that the pursuer can measure the bearing angle $\delta_{P}$
between its own velocity vector and the LOS to the evader, rendering
the following measurement equation:
\begin{equation}
	z = \delta_{P}+\nu = \gamma_{P}-\lambda+\nu
	\label{eq:Meas}
\end{equation}
where $\nu$ is the (possibly non-Gaussian) measurement noise.

\subsection{Linearized Model Used for Guidance}
Differential game-based guidance laws are commonly derived using linear models, and we use such a model to derive this type of guidance law in Sec.~\ref{sec:GL}. Figure~\ref{fig:Linearized_planar_engagement_geometry} presents a schematic view of the linearized planar geometry of the pursuer and the evader.
\begin{figure}[tbh]
  \centering
  \includegraphics[width=3.25in]{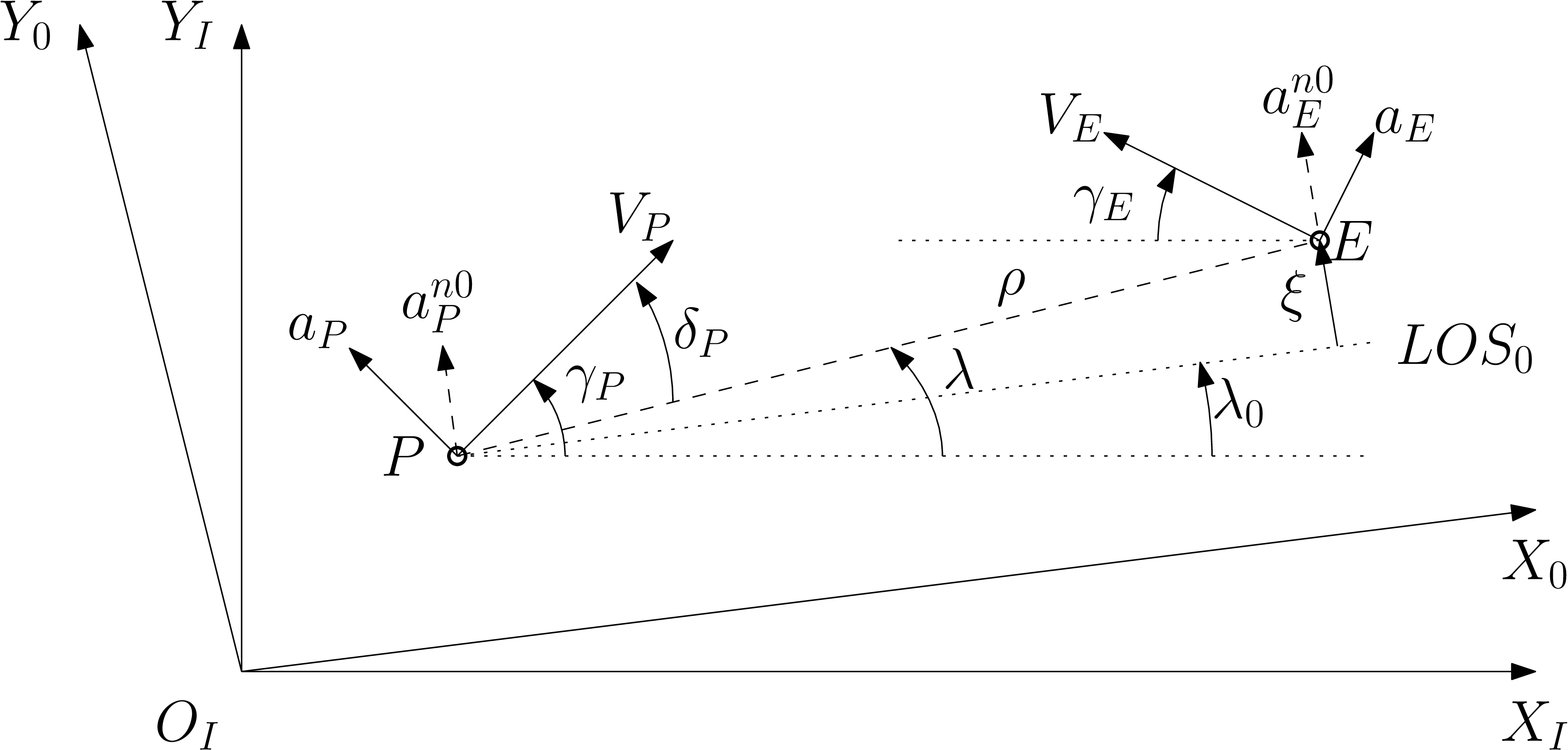}
		\par
	\caption{Linearized planar engagement geometry.}
	\label{fig:Linearized_planar_engagement_geometry}
\end{figure}
The initial LOS, denoted as $LOS_{0}$, is used for the linearization, where the axis $X_{0}$ is parallel to it and the $Y_{0}$ axis is perpendicular to it. We denote the relative separation normal to the LOS by $\xi$.  The accelerations normal to $LOS_{0}$ of the pursuer and the evader, respectively, are denoted by $a_{P}^{n0}$ and $a_{E}^{n0}$.  Assuming small angles, the linearization about the collision course holds, rendering 
\begin{gather}
	a_{P}^{n0} = a_{P}\cos \delta_{P_{0}} \approx a_{P}\\
	a_{E}^{n0} = a_{E}\cos \delta_{E_{0}} \approx a_{E}.
\end{gather}
The state vector of the linearized pursuit--evasion problem is
\begin{equation}
	\textbf{x} = 
	\begin{bmatrix}
		\xi & \dot{\xi} & a_{P} & a_{E}
	\end{bmatrix}^{T}
\end{equation}
and the corresponding linearized EOMs are
\begin{subequations}
	\begin{align}
		\dot{x}_{1} & = x_{2},                           & x_{1}(0) = 0      \;\;\; \\
		\dot{x}_{2} & = x_{4} - x_{3} ,                  & x_{2}(0) = x_{20} \\
		\dot{x}_{3} & = (a_{P}^{c} - x_{3}) / \tau_{P},  & x_{3}(0) = 0      \;\;\; \\
		\dot{x}_{4} & = (a_{E}^{c} - x_{4}) / \tau_{E} , & x_{4}(0) = 0      \;\;\;
	\end{align}
\end{subequations}
where
\begin{equation}
	x_{20} = V_{E} \sin \gamma_{E_{0}} - V_{P} \sin \gamma_{P_{0}}.
	\label{eq:x20_dim}
\end{equation}
The time-to-go is approximated by
\begin{equation}
	t_{go} \dfn t_{f} - t \approx - \rho / V_{\rho}
	\label{eq:t_go}
\end{equation}
where $t_{f}$ denotes the final (interception) time.

Following common practice, we transform the game into nondimensional
form.  The normalized time-to-go is %\YO{(but the second eqn on the line is NOT time to go)}
\begin{equation}
	\tau = t_{go} / \tau_{P}, 
\end{equation}
the nondimensional state vector is
\begin{equation}
	\bar{\textbf{x}} = 
	\begin{bmatrix}
		\frac{x_{1}(t_{f} - \tau_{P}\tau)}{a_{E}^{\max} \tau_{P}^{2}} & 
		\frac{x_{2}(t_{f} - \tau_{P}\tau)}{a_{E}^{\max} \tau_{P}} & 
		\frac{x_{3}(t_{f} - \tau_{P}\tau)}{a_{P}^{\max}} & 
		\frac{x_{4}(t_{f} - \tau_{P}\tau)}{a_{E}^{\max}}
	\end{bmatrix}^{T},
\end{equation}
and the controls become
\begin{subequations}
	\begin{align}
		\bar{u}(\tau) & = u(t_{f} - \tau_{P}\tau), \qquad \abs{\bar{u}(\tau)} \leq 1 \\
		\bar{v}(\tau) & = v(t_{f} - \tau_{P}\tau), \qquad \abs{\bar{v}(\tau)} \leq 1. \label{eq:v_const}
	\end{align}
	\label{eq:cont_const}
\end{subequations}
These normalized states satisfy the following normalized EOM:
\begin{subequations}
	\begin{align}
		d\bar{x}_{1}/d\tau & = - \bar{x}_{2}, & \bar{x}_{1}(\tau_{0}) = 0 \;\;\; \\
		d\bar{x}_{2}/d\tau & = - \bar{x}_{4} + \mu \bar{x}_{3} , & \bar{x}_{2}(\tau_{0}) = \bar{x}_{20} \\
		d\bar{x}_{3}/d\tau & = \bar{x}_{3} - \bar{u},  & \bar{x}_{3}(\tau_{0}) = 0 \;\;\; \\
		d\bar{x}_{4}/d\tau & = (\bar{x}_{4} - \bar{v}) / \epsilon ,  & \bar{x}_{4}(\tau_{0}) = 0 \;\;\;
	\end{align}
\end{subequations}
where $\tau_{0} = t_{f} / \tau_{P}$ is the normalized final time, $\mu = a_{P}^{\max} / a_{E}^{\max}$, and
$\epsilon =\tau_{E} / \tau_{P} $. 
The terms $\mu$ and $\mu \epsilon$ are commonly referred to as the maneuverability ratio and the agility ratio, respectively. \rev{The normalized counterpart of $x_{20}$ in Eq.~\eqref{eq:x20_dim} is $\bar{x}_{20} = (V_{E}\sin\gamma_{E_{0}} - V_{P}\sin\gamma_{P_{0}})/(a^{\max}_{E}\tau_{P})$.}
\rev{%
\subsection{Performance Criterion}
\label{sec:criterion}
The interception is formulated as a zero-sum pursuit--evasion differential game, with the miss distance taken as the cost. The pursuer minimizes, and the evader maximizes:
\begin{equation}
	J = \abs{x_{1}(t_{f})}.
	\label{eq:cost_dim}
\end{equation}
Under the linearized dynamics and the time-to-go approximation~\eqref{eq:t_go}, the cost reduces to $J = \abs{\bar{x}_{1}(0)}$ in the normalized game-coordinate variables.
}%

\section{DGLCC with Time-Varying Delays}
\label{sec:GL}
In this section, we develop a differential game-based guidance law that
can handle two \emph{time-varying} estimation delays: 1) in the
evader's acceleration and 2) in the relative velocity perpendicular to
the LOS. Assuming, for now, that the aforementioned delays
are known, our development follows the derivation of the DGLCC law for
time-invariant delays in~\cite{glizer_linear_2009}.  We first present
the game formulation for the time-varying delayed-information
case. Then, we derive the necessary conditions for the optimal
solution of the game. Finally, we present the game space decomposition
yielding the optimal closed-loop controls. 
For completeness, the solution of the delay-free case, resulting in the DGL1 law, is concisely summarized in
\ifsupp Appendix~\ref{sec:DGL1}\else the Supplementary Material\fi.
We note, in passing,
that the guidance law developed herein is to be driven by suitably
timed estimates, as will be shown in the sequel.
\subsection{Time-Varying Delayed Information Game}
Incomplete and noisy measurements generate a dynamic estimation error, resulting in uncertainty about the target's evasion maneuver, which renders the perfect-information assumption unrealistic. A standard method of coping with this uncertainty is to assume that the pursuer has access only to a time-delayed estimate, promoting the following definition of the pursuer's information state vector (in a normalized setting):
\begin{equation}
	\bar{\textbf{w}}(\tau) = 
	\begin{bmatrix}
		\bar{x}_{1}(\tau) & \bar{x}_{2} (\tau + \Delta_{1}(\tau)) & \bar{x}_{3}(\tau) & \bar{x}_{4} (\tau + \Delta_{2}(\tau))
	\end{bmatrix}^{T}
\end{equation}
where $\Delta_{1}(\tau)$ and $\Delta_{2}(\tau)$ are the information time delays of the relative velocity and the evader acceleration, respectively. We also assume that $0<\Delta_{1}(\tau) \leq \Delta_{2}(\tau)$ for any $\tau\in[0,\tau_{0}]$, and that both functions are known, monotonically increasing, continuous for $\tau \geq 0$, and differentiable for $\tau > 0$.

In the perfect-information case, the optimal controls depend on the
zero-effort miss (ZEM), $\bar{z}(\tau)$. % ~\eqref{eq:z_DGL1}.  
% \begin{equation}
% 	\bar{z}(\tau) = \bar{x}_{1}(\tau) + \tau \bar{x}_{2}(\tau) - \mu \Psi(\tau)\bar{x}_{3}(\tau) + \epsilon^{2} \Psi (\tau / \epsilon) \bar{x}_{4}(\tau).
% 	\label{eq:z_DGL1}
% \end{equation}
When the information is
delayed, the ZEM at the current time becomes an uncertainty set,
$\bar{Z}(\tau)$, consisting of all possible values of $\bar{z}(\tau)$
given the information vector $\bar{w}$.  Define
\begin{gather}
	\bar{z}_{1}(\tau) \dfn \bar{x}_{1}(\tau) - \mu
	\Psi(\tau)\bar{x}_{3}(\tau)\label{eq:z_1}\\
	\bar{z}_{2}(\tau) \dfn \tau \bar{x}_{2}(\tau) + \epsilon^{2} \Psi
	(\tau / \epsilon) \bar{x}_{4}\rev{(\tau)}
	\label{eq:z_2}  
\end{gather}
where $\bar{x}_{1}(\tau)$ and $\bar{x}_{3}(\tau)$ are assumed to be perfectly known and delay-free, whereas $\bar{x}_{2}(\tau)$ and $\bar{x}_{4}(\tau)$ are assumed to be perfectly known but delayed, and 
\begin{equation}
	\Psi(\tau) = \exp(-\tau) + \tau - 1.
    \label{eq:Psi}
\end{equation}
With these definitions on hand, the uncertainty set is expressed as
\begin{equation}
	\bar{Z}(\tau) = [\bar{z}_{1}(\tau) + m(\tau),\bar{z}_{1}(\tau) + M(\tau)]
\end{equation}
where $m(\tau)$ and $M(\tau)$ are the solutions of the following optimization problems:
\begin{subequations}\label{eq:opt_mM}
	\begin{align}
		m(\tau) &= \min_{\abs{\bar{v}(\sigma)}\leq1,\;\sigma\in[\tau,\tau + \Delta_{2}(\tau)]} \bar{z}_{2}(\tau)  \\
		M(\tau) &= \max_{\abs{\bar{v}(\sigma)}\leq1,\;\sigma\in[\tau,\tau + \Delta_{2}(\tau)]} \bar{z}_{2}(\tau).
	\end{align}
\end{subequations}
\rev{%
Computing the center of the uncertainty set is the conceptual cornerstone of the game's reformulation. Since $\bar{z}_{1}(\tau)$ is delay-free and perfectly known to the pursuer [Eq.~\eqref{eq:z_1}], the structure of $\bar{Z}(\tau)$ is fully determined by the achievable range of $\bar{z}_{2}(\tau)$, which depends, through the dynamics of $\bar{x}_{2}$ and $\bar{x}_{4}$, on the evader's control trajectory $\bar{v}$ over the delay window $[\tau, \tau + \Delta_{2}(\tau)]$. As the pursuer has no further information that would bias its choice within $\bar{Z}(\tau)$, it is natural to take the center of $\bar{Z}(\tau)$ as the working state of the game. By the definition of the uncertainty set, this center is $\bar{z}_{1}(\tau)$ plus the midpoint of the achievable range of $\bar{z}_{2}(\tau)$, namely, $[m(\tau) + M(\tau)]/2$.}

\rev{%
The derivation of this midpoint, given in full in
\ifsupp Appendix~\ref{append1}\else the Supplementary Material\fi,
proceeds in three steps. First, the admissible range of $\bar{x}_{4}$ at the intermediate time $\tau + \Delta_{1}(\tau)$ is determined by integrating the dynamics on $[\tau + \Delta_{1}(\tau), \tau + \Delta_{2}(\tau)]$ with boundary condition $\bar{x}_{4}(\tau + \Delta_{2}(\tau)) = \bar{w}_{4}(\tau)$; the bounds depend monotonically on $\bar{v}$ over this subinterval. Second, $\bar{z}_{2}(\tau)$ is expressed as an affine functional of $\bar{v}$ on $[\tau, \tau + \Delta_{1}(\tau)]$ and of the intermediate value $\bar{x}_{4}(\tau + \Delta_{1}(\tau))$. Third, since this expression is linear in both arguments, the extrema $m(\tau)$ and $M(\tau)$ are attained at $\bar{v} \equiv \mp 1$, and the symmetric structure of the resulting expressions yields the closed-form center given below:
}%
\begin{align}
	\bar{z}_{cc}(\tau) =\; & \bar{x}_{1}(\tau) + \tau \bar{x}_{2}(\tau + \Delta_{1}(\tau)) - \mu \Psi(\tau)\bar{x}_{3}(\tau) - \mu \tau \int_{\tau}^{\tau + \Delta_{1}(\tau)} \bar{x}_{3}(s)ds \notag \\
	& + \epsilon \exp(-\Delta_{2}(\tau)/\epsilon)[\tau
	\exp(\Delta_{1}(\tau)/\epsilon) + \epsilon\exp(-\tau/\epsilon) -
	\epsilon] \bar{x}_{4}(\tau + \Delta_{2}(\tau)).
	\label{eq:z_cc}
\end{align}

Next, we seek to find the evolution equation for $\bar{z}_{cc}(\tau)$.
Differentiating Eq.~\eqref{eq:z_cc} yields 
(see \ifsupp Appendix~\ref{append2}\else the Supplementary Material\fi\ for the detailed derivation)
\begin{align}
  \frac{d \bar{z}_{cc}(\tau)}{ d \tau} & = \mu \Psi(\tau) \bar{u}(\tau) - \int_{\tau}^{\tau+\Delta_{1}(\tau)}\bar{v}(s)ds +  \int_{\tau}^{\tau+\Delta_{2}(\tau)}\exp((\tau-s)/\epsilon)\bar{v}(s)ds \notag \\
  - & \exp(\Delta_{1}(\tau)/\epsilon) [\tau(1+\gamma_{1}(\tau))/\epsilon + 1]  \int_{\tau + \Delta_{1}(\tau)}^{\tau + \Delta_{2}(\tau)} \exp((\tau - s)/\epsilon) \bar{v}(s)ds  \notag \\
  - &\exp(-\Delta_{2}(\tau)/\epsilon)[\tau \exp(\Delta_{1}(\tau)/\epsilon) + \epsilon\exp(-\tau/\epsilon) - \epsilon] \bar{v}(\tau + \Delta_{2}(\tau))(1 + \gamma_{2}(\tau)) 
\label{eq:z_cc_star}
\end{align}
where
\begin{equation}
\gamma_{1}(\tau) \dfn d \Delta_{1}(\tau) / d\tau, \qquad
\gamma_{2}(\tau) \dfn d \Delta_{2}(\tau) / d\tau.
\end{equation}

\begin{remark}
  Explicitly addressing time-dependent delays, the dynamic
  equation for $\bar{z}_{cc}$, Eq.~\eqref{eq:z_cc_star}, differs from
  the corresponding equation in~\cite{glizer_linear_2009} but reduces
  to it when the delays are set to be time-invariant.
\end{remark}

This dynamics is independent of $\bar{z}_{cc}$ and the state variables $\bar{x}_{i}$, $i=1,...,4$. It depends only on the pursuer's (delay-free) and evader's (time-delayed) controls.
We can formulate a new game by using $\bar{z}_{cc}$ and its dynamics, which we rewrite as
\begin{equation}\label{eq:z_cc_dyn}
	d \bar{z}_{cc}/d\tau = \mu \Psi(\tau) \bar{u}(\tau) + \mathcal{F}(\bar{v}_{\tau}(\cdot),\tau)
\end{equation}
where
\begin{equation}
\bar{v}_{\tau}(\sigma) = \bar{v}(\tau + \sigma)
\end{equation}
for each $\tau \in [0,\tau_{0}]$ and $\sigma \in [0,\Delta_{2}(\tau)]$. We define $\mathcal{F}$ as
\begin{align}\label{eq:F_def}
\mathcal{F}&(\bar{v}_{\tau}(\cdot),\tau)  = - \int_{\tau}^{\tau+\Delta_{1}(\tau)}\bar{v}(s)ds +  \int_{\tau}^{\tau+\Delta_{2}(\tau)}\exp((\tau-s)/\epsilon)\bar{v}(s)ds \notag \\
& - \exp(\Delta_{1}(\tau)/\epsilon) [\tau(1+\gamma_{1}(\tau))/\epsilon + 1]  \int_{\tau + \Delta_{1}(\tau)}^{\tau + \Delta_{2}(\tau)} \exp((\tau - s)/\epsilon) \bar{v}(s)ds  \notag \\
& - \exp(-\Delta_{2}(\tau)/\epsilon)[\tau \exp(\Delta_{1}(\tau)/\epsilon) + \epsilon\exp(-\tau/\epsilon) - \epsilon] \bar{v}(\tau + \Delta_{2}(\tau))(1 + \gamma_{2}(\tau)).
\end{align}

Since the evader control $\bar{v}$ is delayed, we add the initial condition to the game formulation
\begin{equation}
\bar{v}(\zeta) = \varphi(\zeta), \quad \zeta \in (\tau_{0},\tau_{0} + \Delta_{2}(\tau_0)]
\end{equation}
where $\varphi$ satisfies
\begin{equation}
	\abs{\varphi(\zeta)} \leq 1, \quad \zeta \in (\tau_{0},\tau_{0} + \Delta_{2}(\tau_0)]
\end{equation}
which is the same control constraint as in Eq.~\eqref{eq:cont_const}.

\rev{The new game's cost function is
\begin{equation}
    J_{cc} = \abs{\bar{z}_{cc}(0)} = \abs{\bar{z}(0)} =
    \abs{\bar{x}_{1}(0)} = J.
  \label{eq:cost}
\end{equation}
This cost is minimized by the pursuer through the running
  control $\bar{u}(\cdot)$ on $[0,\tau_{0}]$, and maximized by the
  evader through both the running control $\bar{v}(\cdot)$ on
  $[0,\tau_{0}]$ and the delayed-extension control $\varphi(\cdot)$ on
  $(\tau_{0},\tau_{0}+\Delta_{2}(\tau_{0})]$, yielding the
  saddle-point value
\begin{equation}
	J_{cc}^{*} = \min_{\bar{u}}\max_{(\bar{v},\varphi)}\, J_{cc}.
	\label{eq:cost_saddle}
\end{equation}}
\subsection{Necessary Conditions for Optimality}
\begin{theorem}
If $\bar{u}^{*}(\tau)$ and $\bar{v}^{*}(\tau)$ are the optimal controls, and $\varphi^{*}(\zeta)$ is the optimal initial condition of the delayed-information differential game, then
\begin{subequations}
\begin{align}
\bar{u}^{*}(\tau) &= \sign \bar{z}_{cc}(0), \qquad \tau \in [0,\tau_{0}] \\
\bar{v}^{*}(\tau) &= \sign \bar{z}_{cc}(0), \qquad \tau \in [0,\tau_{0}] \\
\varphi^{*}(\zeta)  &= \sign \bar{z}_{cc}(0), \qquad \zeta \in (\tau_{0},\tau_{0} +\Delta_{2}(\tau_{0})]
\end{align}
\end{subequations}
\label{th:1}
\end{theorem}
% The proof is deferred to the supplementary material due to the paper's length limitation.
\begin{proof}
  The proof is deferred to \ifsupp Appendix~\ref{append3}\else the Appendix\fi.
\end{proof}

\rev{%
\begin{remark}
Theorem~\ref{th:1} establishes the saddle point of the game, characterizing the optimal strategies of both players. The formulation is asymmetric in two respects. First, in the information structure: following the worst-case assumption standard in the differential game literature~\cite{shinar_solution_1999,glizer_linear_2009}, the pursuer has access only to delayed estimates of $\bar{x}_{2}$ and $\bar{x}_{4}$, whereas the evader is assumed to have perfect information on the full state, in addition to knowing the pursuer's information structure; this is the most adversarial setting from the pursuer's standpoint. Second, in the control spaces: due to the delayed information, the evader's strategy comprises both the running control $\bar{v}(\tau)$ on $[0,\tau_{0}]$ and the delayed-extension control $\varphi(\zeta)$ on $(\tau_{0},\tau_{0}+\Delta_{2}(\tau_{0})]$, whereas the pursuer's control acts only on $[0,\tau_{0}]$. Despite these asymmetries, the optimal controls of both players reduce to the same constant value, $\sign\bar{z}_{cc}(0)$, over their respective domains.
\end{remark}
}%
\subsection{Game Space Decomposition}
\rev{Substituting the optimal controls of Theorem~\ref{th:1} into the dynamics~\eqref{eq:z_cc_dyn} yields the closed-loop dynamics of the optimal trajectories. As established in stage~1 of the proof of Theorem~\ref{th:1} (see \ifsupp Appendix~\ref{append3}\else the Appendix\fi), the delayed-control functional $\mathcal{F}$ admits the bound
\begin{equation}
	\abs{\mathcal{F}(\bar{v}_{\tau}(\cdot),\tau)} \leq A(\tau), \qquad \tau \in [0,\tau_{0}],
	\label{eq:F_bound_main}
\end{equation}
with $A(\tau)$ defined in Eq.~\eqref{eq:A_tau}; the bound follows from the negativity of the integrands and of the coefficient of $\bar{v}_{\tau}(\Delta_{2}(\tau))$ in the expression for $\mathcal{F}$. Consequently, the bound is attained at the constant-sign extrema $\bar{v}_{\tau}(\sigma) \equiv \pm 1$: by direct substitution, $\bar{v}_{\tau}(\sigma) \equiv +1$ yields $\mathcal{F} = -A(\tau)$, while $\bar{v}_{\tau}(\sigma) \equiv -1$ yields $\mathcal{F} = +A(\tau)$. Substituting the optimal evader controls $\bar{v}^{*}_{\tau}(\sigma) = \sign\bar{z}_{cc}(0)$ thus gives $\mathcal{F}(\bar{v}^{*}_{\tau}(\cdot),\tau) = -A(\tau)\sign\bar{z}_{cc}(0)$. Combining this with the optimal pursuer control reduces Eq.~\eqref{eq:z_cc_dyn} to}
\begin{equation}
	d \bar{z}_{cc}/d\tau = R(\tau) \sign \bar{z}_{cc}(0)
	\label{eq:zcc_dyn_opt}
\end{equation}
where
\begin{equation}
	R(\tau) \dfn \mu \Psi(\tau) -  A(\tau) 
	\label{eq:R_def}
\end{equation}
and 
\begin{align}
	A(\tau) &\dfn  \Delta_{1}(\tau) + \tau(1+\gamma_{1}(\tau)) + \epsilon
	\exp(-(\tau + \Delta_{2}(\tau))/\epsilon) (1 + \gamma_{2}(\tau))  \notag \\
	- &  \epsilon\exp( - \Delta_{2}(\tau) / \epsilon )\gamma_{2}(\tau) 
	+  \exp((\Delta_{1}(\tau) - \Delta_{2}(\tau))/\epsilon)[\tau (\gamma_{2}(\tau) - \gamma_{1}(\tau)) - \epsilon]
    \label{eq:A_tau}
\end{align}
and it is shown that $A(\tau) > 0 $ for any $\tau \in [0 , \tau_{0}]$ in 
% the proof of Thm.~\ref{th:1}. 
\ifsupp Appendix~\ref{append3}\else the Appendix\fi.

The behavior of optimal trajectories of the game is determined by the order and the number of sign changes of the function $R(\tau)$. As this function depends on the time delays, we assume the following functional forms for the time-varying delays~\cite{shinar_linear_2002}:
\begin{subequations}
	\label{eq:delta_form}
	\begin{align}
		\Delta_{1}(\tau) &= a + b_{1} \tau^{\omega} \\
		\Delta_{2}(\tau) &= a + b_{2} \tau^{\omega}
	\end{align}
\end{subequations}
where $0 < a$, $0 \leq b_{1} \leq b_{2}$, and $0 < \omega < 1$ are given constants. 
The motivation for this functional form, and in particular the value of $\omega$, will be provided in Sec.~\ref{sec:simp_mod}, based on a first-principles analysis of the detection delay.
We next wish to investigate the properties of $R(\tau)$. Before doing so, we make the assumption that the pursuer's maneuver capability is superior to that of the evader, that is, $\mu > 1$. This assumption is reasonable, as the evader is assumed to maneuver effectively, and it is well known that the pursuer's options are very limited against an optimal evasion maneuver unless it possesses a maneuverability advantage over the evader~\cite{shinar_solution_1981}. With this assumption on hand, we have the following lemma.
\begin{lemma}
	For $\mu > 1$, the function $R(\tau)$ has at least one positive root, $\tau_{s}$, for any $0 < a$, $0 < b_{1} \leq b_{2}$, and $0 < \omega < 1$.
\end{lemma}
\begin{proof}
	We have that
	\begin{equation}
		R(0) = - A(0) < 0
	\end{equation}
	and 
	\begin{equation}
		\lim\limits_{\tau \rightarrow \infty} R(\tau) \approx (\mu - 1) \tau > 0
	\end{equation}
	Hence, there exists $\tau'$ such that $R(\tau) > 0$ for all
        $\tau\geq \tau'$.  Noting that $R(\tau)$ is a continuous
        function yields the lemma.
\end{proof}
Moreover, we conjecture that $R(\tau)$ has a single positive root
$\tau_{s}$. Although not formally proven, this conjecture is supported
by the results of a thorough simulation study, comprising the
following $1$,$197$,$900$ cases (for brevity, we use the
MATLAB\textsuperscript{\textregistered} range creation notation):
\begin{equation}
\begin{aligned}
	a &= 0.01:0.01:0.1, & b_{2} = 0.06:0.06:0.6,\quad b_{1} &= (0:0.125:1) b_{2}, \\
	\rev{\omega} &= (1:1:11)/12,   & \mu = 1.1:0.19:3, \;\;\qquad \epsilon &=  0.2:0.18:2.
\end{aligned}
\end{equation}

Figure~\ref{fig:Game_Space_mu_1} presents a schematic view of the
decomposition of the game space when $R(\tau)$ has a single positive
root. The game space is decomposed into a regular region, denoted by
$\mathcal{D}_{1}$, and a singular region, denoted by
$\mathcal{D}_{0}$. These regions are determined by
\begin{subequations}
	\begin{align}
		\mathcal{D}_{1} &= \{ 0\leq \tau \leq \tau_{s} \} \cup \left\lbrace \abs{\bar{z}_{cc}} \geq \int_{\tau_{s}}^{\tau} R(s)ds, \quad \tau \geq \tau_{s} \right\rbrace \\
		\mathcal{D}_{0} &= \left\lbrace \abs{ \bar{z}_{cc}} < \int_{\tau_{s}}^{\tau} R(s)ds, \quad \tau \geq \tau_{s} \right\rbrace.
	\end{align}
\end{subequations}

In the regular region, $\mathcal{D}_{1}$, the optimal feedback strategies of both players are~\rev{\cite{glizer_linear_2009,shima_time-varying_2002}}
\begin{equation}
	\bar{u}^{*}(\tau) = \bar{v}^{*}(\tau) = \sign \bar{z}_{cc}(\tau)
	\label{eq:opt_feedback}
\end{equation}
and the value of the game in this case is
\begin{equation}
	J_{cc}^* = \abs{\bar{z}_{cc}(\tau_{0})} + \int_{\tau_{0}}^{0} R(\tau)d\tau
	\label{eq:J1}
\end{equation}
which never vanishes.

In the singular region $\mathcal{D}_{0}$, the optimal strategies of both players are arbitrary. The corresponding value of the game is
\begin{equation}
	J_{cc}^* = \int_{\tau_{s}}^{0} R(\tau)d\tau \dfn M_{s}
	\label{eq:J2}
\end{equation}
which never vanishes as well. For practical reasons, we use a linear,
chatter-prevention strategy for the pursuer inside the singular
region~\rev{\cite{shaferman_stochastic_2016}}, yielding
\begin{equation}
	\bar{u}^{*}(\tau) = 
	\begin{cases}
		\sign \bar{z}_{cc}(\tau), & \bar{z}_{cc} \in \mathcal{D}_{1} \\
		\text{sat} \frac{\bar{z}_{cc}(\tau)}{k \bar{z}^{*}_{cc}(\tau)}, & \bar{z}_{cc} \in \mathcal{D}_{0} 
	\end{cases}
	\label{eq:DGLCC_GL}
\end{equation}
where $\sat(\cdot)$ stands for the saturation function, $0 < k \leq 1$ is the portion of the singular region in which the control is linear, and $\bar{z}^{*}_{cc}$ is the
boundary of the singular region, computed as
\begin{equation}
    \bar{z}^{*}_{cc}(\tau) =
    \begin{cases}
        \int_{\tau_{s}}^{\tau} R(s)ds, & \tau \geq \tau_s \\
        0, & 0\leq \tau < \tau_s
    \end{cases}.
\end{equation}

\begin{figure}[tbh]
  \centering
  \includegraphics[width=3.25in]{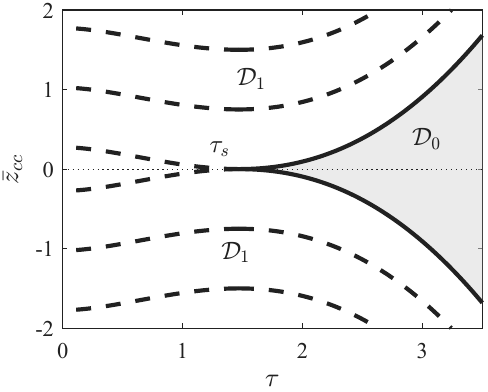}
  \caption{Game space, $\mu > 1$.}
  \label{fig:Game_Space_mu_1}
\end{figure}

\begin{remark}
  The optimal closed-loop controls are calculated based on the center
  of the uncertainty set,
  Eq.~\eqref{eq:z_cc}, which requires time-delayed information about
  the states $\bar{x}_{2}$ and $\bar{x}_{4}$.  The standard approach,
  employed in the existing literature, is to use the current-step
  estimates (i.e., the filtered estimates), effectively regarding them
  as uniformly delayed throughout the
  game~\cite{shinar_nonorthodox_2002,glizer_linear_2009}.  This
  approach can only be regarded as a rough approximation, at best. In
  contradistinction, we drive our guidance law with estimates
  corresponding to precisely computed delayed time, as required by its
  derivation. As presented in the next section, such delayed estimates
  are computed using a state smoother.
\label{rem:smooth}
\end{remark}

\section{Particle Filtering and Smoothing for Delayed Interception}
\label{sec:est}
The guidance law developed in Sec.~\ref{sec:GL} requires a
properly timed estimate of the state, along with an estimate of the
inherent time delay associated with this estimate. We do this by using
the following tools. The bulk of the estimation workload is performed
by an interacting multiple-model particle filter
(IMMPF)~\cite{blom_exact_2007}. We conjoin the IMMPF with a
semi-Markov transition mechanism, which is used to estimate the
uncertainty interval following a target evasive maneuver and the
associated time delays. Finally, we properly feed the guidance law
with a delayed state estimate, which is computed (given the estimated
time delays) using a fixed-lag particle-based smoother.
\subsection{IMMPF Algorithm}
The IMMPF algorithm~\cite{blom_exact_2007} is a multiple-model
sequential MC estimator that can address problems characterized by
nonlinear dynamics, non-Gaussian driving noises, and multiple models
with non-Markovian switching modes. This algorithm generally follows
the logic of the interacting multiple-model (IMM)
estimator~\cite{blom_interacting_1988}; in the IMMPF, the resampling
stage is combined with the interaction stage before the prediction and
filtering stages. The filter runs a bank of particle filters (PFs)
matched to all possible modes.  At each time $t_{k}$, the probability
density function (PDF) is represented by the set of particles and
associated weights
$\{ \textbf{x}_{k}^{r,s}, w_{k}^{r,s} \mid r\in\mathbb{M}, s\in \{1,
\dots, S\}\} $, where $\textbf{x}_{k}^{r,s}$ is the state vector of
particle $s$ and mode $r$, and $w_{k}^{r,s}$ is its associated weight.
$\mathbb{M}$ is the set of $M$ discrete modes, and $S$ is the number
of particles of each mode, rendering the total number of particles
$N_{p}=MS$.

\subsection{Uncertainty Interval Estimation}
For insight and context, we first review an existing offline
analytical framework commonly used to model the generation of evasion
maneuver-caused detection delays. We then present a probabilistic
model of the uncertainty interval and use it to develop a method for
estimating this interval using semi-Markov modeling of the target's
evasion maneuvers.
\subsubsection{Simplified Modeling of the Detection Delay}
\label{sec:simp_mod}
Depicted in Fig.~\ref{fig:evas_man}, Hexner
et~al.~\cite{hexner_temporal_2008} proposed a simplified physical model
for the generation of time delay in a linearized system driven by
Gaussian measurement noises.  The model assumes that the evader has
first-order autopilot dynamics, but it can be easily modified to suit
any linear time-invariant system of arbitrary
order~\cite{shaferman_near-optimal_2021}. In Fig.~\ref{fig:evas_man},
the change in the evasion maneuver, $\Delta a_{E}^{c}$, goes through
the evader's autopilot dynamics and then through a double integrator,
which translates the estimation error of the evader acceleration
normal to the LOS, $\tilde{a}_{E}^{n0}$, to the estimation error for
the relative separation perpendicular to the LOS, $\tilde{\xi}$.
\begin{figure}[tbh]
  \centering \includegraphics[width=3.25in]{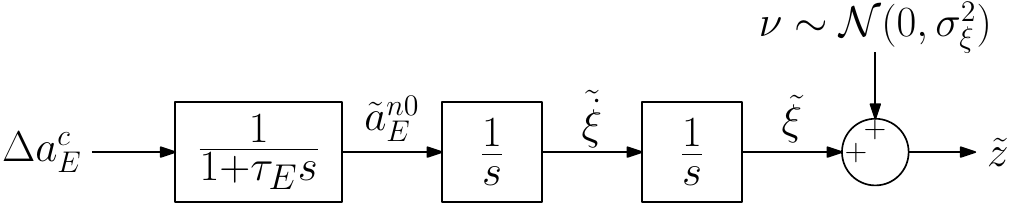}
  \caption{Simplified physical model for estimation delay
    generation~\cite{hexner_temporal_2008}.}
  \label{fig:evas_man}
\end{figure}

The sudden evasion maneuver changes are modeled using a step response,
which translates to the following approximation, holding for any brief
period $ t_{r} = t - t_{sw} $ after the sudden evasion change:
\begin{subequations}
\label{eq:HWD_app}
	\begin{align}
		\tilde{a}_{E}^{n0}(t_{r}) & = \Delta a_{E}^{c} [1 - \exp(-t_{r} / \tau_{E})] \approx \frac{\Delta a_{E}^{c} t_{r}}{ \tau_{E}}, \\
		\tilde{\dot{\xi}}(t_{r})  & \approx \frac{\Delta a_{E}^{c} t_{r}^{2}}{2\tau_{E}},                                               \\
		\tilde{\xi}(t_{r})        & \approx \frac{\Delta a_{E}^{c} t_{r}^3}{6\tau_{E}}.
	\end{align}
\end{subequations}
Setting a threshold, $k_{\xi}$, for the number of standard deviations
required for the estimator to detect the maneuver change then yields
the estimation error at detection:
\begin{equation}
	\tilde{\xi} = k_{\xi} \sigma_{\xi},
\end{equation}
where a common practice is to set
$k_{\xi} = 2$~\cite{hexner_temporal_2008}. Denoting by $\Delta t_{r}$
the size of the interval beginning at maneuver change and ending at
maneuver detection, we thus have
\begin{equation}
\Delta t_{r} \approx \left( \frac{6 \tau_{E} k_{\xi} \sigma_{\xi}}{\Delta a_{E}^{c}} \right)^{1/3}.
\end{equation}
Assuming small deviations from the collision course yields the
approximation
$\xi \approx (\lambda - \lambda_{0})
\rho$~\cite{shaferman_stochastic_2016}. Furthermore, realistically
assuming a relatively small range estimation error yields
\begin{equation}
	\sigma_{\xi} \approx \rho \sigma_{\lambda}
\end{equation}
where $\sigma_{\lambda}$ is the standard deviation of the bearing
measurement noise, Eq.~\eqref{eq:Meas}. This finally yields the
approximation
\begin{equation}
	\Delta t_{r} \approx \left( \frac{6 \tau_{E} k_{\xi} \rho \sigma_{\lambda}}{\Delta a_{E}^{c}} \right)^{1/3}
	\approx \left( - \frac{6 \tau_{E} k_{\xi}  V_{\rho} t_{go} \sigma_{\lambda}}{\Delta a_{E}^{c}} \right)^{1/3}
	\label{eq:theta_approx}
\end{equation}
where the time-to-go approximation from Eq.~\eqref{eq:t_go} is used in
the second transition.

The approximation~\eqref{eq:theta_approx} provides some essential
insights regarding the uncertainty interval: 1) it is time-varying, 2)
it can be modeled as
\begin{equation}
	\Delta (t_{go}) = a + b t_{go}^{1/3},
	\label{eq:delta_model}
\end{equation}
implying that $\omega = 1 / 3$ for the models in
Eq.~\eqref{eq:delta_form}. Furthermore, we set $a = 1 / f$ because if
the evader changes its maneuver at the final time step of the
engagement, the estimator cannot detect that, as the effect of this
change is negligible.

The approximation of the uncertainty interval, as presented above,
suffers from two notable shortcomings. First, it completely disregards
the information acquired in real time throughout the
engagement. Second, it is only valid for linear models and Gaussian
noise distributions. In what follows, we propose a novel approach for
estimating the uncertainty interval that overcomes both shortcomings
by using the outputs of the IMMPF algorithm. We commence the
exposition by defining the uncertainty interval in mathematical terms.
\subsubsection{Probabilistic Modeling of the Uncertainty Interval}
\label{sec:prob_uncertain}
Consider the interval $[t, t_k]$, where $t_k$ is the present time.
For any $t_k$, the \emph{uncertainty interval}, denoted by
$[t^\star,t_k]$, is determined by its lower bound, $t^\star$. This
bound is defined as the time for which any evasive maneuver executed
by the target at $t<t^\star$ has already been detected by the
pursuer's estimator w.p.~1, while an evasive maneuver occurring at
$t\in [t^\star,t_k]$ can only be detected with probability strictly
smaller than 1.

Formally speaking, let $M(t)$ be any continuous random variable
measuring the estimator's maneuver detection vulnerability, such that
the probability that the pursuer's estimator has \emph{not} detected a
target evasive maneuver in the interval $[t_1,t_2]$ is
\begin{equation}
  \label{eq:no_det}
  \ProbM{t_1}{t_2} = \int_{t_1}^{t_2}f_M(\xi)\textrm{d}\xi,
\end{equation}
where $f_M$ is the PDF of $M$.  Notice that
\begin{equation}
  \label{eq:f_M_supp}
  \ProbM{t_1}{t_2} = 0 \iff [t_1,t_2] \cap \supp f_M(t) = \varnothing  
\end{equation}
where $\supp f$ denotes the support of $f$.

Now, according to our model, $\ProbM{t}{t^\star} = 0$ for any
$t < t^\star$, while $\ProbM{t^\star}{t} > 0$ for any
$t \in (t^\star, t_k]$. Along with Eq.~\eqref{eq:f_M_supp}, this then
motivates the following definition of the uncertainty interval:
\begin{equation}
  \label{eq:uncert_def}
  [t^\star,t_k] \dfn \supp f_M(t),
\end{equation}
which also serves as an implied definition of $t^\star$.

\subsubsection{Semi-Markovian Switching Model}
\label{sec:SMC}
To estimate the size of the uncertainty interval, we model the target's
maneuver-switching mechanism as a semi-Markov chain process. We note
that this modeling is admissible, as the IMMPF algorithm is not
limited to strictly Markovian mode-switching mechanisms, thus
permitting its use when the transition probability matrix (TPM) is
state-dependent.  Accordingly, the transition probability from mode
$i$ at time $t_{k-1}$ to mode $j$ at time $t_{k}$, $p_{ij}^{k}$, is
expressed as
\begin{equation}
  p_{ij}^{k} (\theta) = \Pr (r_{k} = j \vert r_{k-1} = i , \theta_{k-1} = \theta)
\end{equation}
where $r_{k}$ denotes the evader's maneuver mode at time $t_{k}$ and
$\theta$ is the sojourn time. The sojourn time at time $t_{k}$,
$\theta_{k}$, represents the elapsed time from the most recent modal
switching by the evader. We note that the notion of the sojourn time
resembles the notion of the elapsed time from the abrupt evasion
change, $t_{r}$, which we have used earlier in
Eq.~\eqref{eq:HWD_app}. This allows us to use the sojourn time to
estimate the uncertainty interval in the sequel.

Adding the sojourn time as an auxiliary variable to the pursuer's
state vector, Eq.~\eqref{eq:StateVec1}, results in the following
augmented state vector:
\begin{equation}
	\textbf{x}_{P} = \begin{bmatrix}
		\rho & \lambda & \gamma_{E} & a_{E} & \theta
	\end{bmatrix}^{T}.
	\label{eq:AugStateVec}
\end{equation} 
Following~\cite{blom_hybrid_1991}, the discrete-time dynamic equation
associated with the sojourn time is 
\begin{equation}
	\theta_{k} = \mathbbm{1}(r_{k},r_{k-1}) \theta_{k-1} + 1 / f
	\label{eq:theta}
\end{equation}
where $f$ is the measurement rate in Hertz, and $ \mathbbm{1}$ is the
indicator random variable,
\begin{equation}
	\mathbbm{1}(r_{k},r_{k-1}) = \begin{cases}
		1 , & r_{k} = r_{k-1} \\ 
		0 , & r_{k}\neq r_{k-1}
	\end{cases}.
\end{equation}
\subsubsection{Estimation of the Uncertainty Interval}
Prior to addressing the problem of uncertainty interval estimation
using the semi-Markov model, we define the notion of mode domination
as follows:
\begin{definition}
  Mode $r$ is said to \emph{dominate} all other modes at time $t_k$ if
  its modal probability exceeds a predetermined threshold,
  $W_{\text{Thres}}$.
\end{definition}
Within the IMMPF framework, the modal probability of a particular mode
is approximated via the sum of the weights of the particles associated
with that mode. Thus, if mode $r$ is dominant, then
\begin{equation}
  \sum_{s=1}^{S} w_{k}^{r,s} > W_{\text{Thres}}
\end{equation}
      
With this definition on hand, we address the problem of uncertainty
interval estimation in two separate cases.
\begin{description}
\item[A dominant mode exists.]  Let mode $r\in\mathbb{M}$ be the
  dominant mode at time $t_k$, and denote the set of all other modes
  by $\Mcr = \mathbb{M}\setminus \{r\}$. In addition, let the PDF of
  the sojourn time, $\theta_k$, assuming that the true mode is a
  member of $\Mcr$, be denoted by
  $f_{\theta_k\mid\Mcr}(\theta)$. Because a sojourn time $\theta$
  (assuming that the true mode is a member of $\Mcr$) means that the
  target has performed, at time $t_k - \theta$, an evasive maneuver
  that has not yet been detected by the estimator (as mode $r$ is
  dominant), the event $\theta_k \in [\alpha,\beta]$ is equivalent to
  the event ``an undetected evasive maneuver has occurred in the time
  interval $[t_k-\beta, t_k-\alpha]$,'' and, consequently, its
  probability can be computed as
  \begin{equation}
    \label{eq:prob_fr}
    \ProbM{t_k-\beta}{t_k-\alpha} =
    \int_{\alpha}^{\beta} f_{\theta_k \mid \Mcr}(\xi)\mathrm{d}\xi
  \end{equation}
  Comparing Eqs.~\eqref{eq:no_det} and~\eqref{eq:prob_fr}, we
  conclude that the PDF $f_{\theta_k \mid \Mcr}$ fulfills the same \rev{role} as $f_M$ in defining the probabilistic characteristics of the
  estimator's detection process. Thus, per Eq.~\eqref{eq:uncert_def},
  the uncertainty interval in this case is $[t_k-\theta^\star_k,t_k]$,
  where $\theta^\star_k$ satisfies
  \begin{equation}
    \label{eq:uncert_def_dom}
    [0, \theta^\star_k] = \supp f_{\theta_k \mid \Mcr}(\theta)
  \end{equation}

  Now, within the IMMPF framework, the PDF $f_{\theta_k \mid \Mcr}$ is
  approximated (via samples and their associated weights) by all
  particles \emph{not} associated with mode $r$ at time $t_k$. Thus,
  observing the sojourn times of all such particles at time $t_k$, an
  estimate of $\theta^\star$ is
\begin{equation}
\label{eq:theta_max_alt}
\hat\theta^\star_k = \max\{\theta_{k}^{r',s} \mid r' \in \Mcr, s= 1,\dots,S\}
\end{equation}
% as this is the maximal possible interval during which a switch to a
% different mode, $r' \neq r$, has occurred. 

To illustrate the idea underlying this approach, consider a two-mode
case, where, at a particular time $t_k$, the modal probability of mode
1 is 0.985, rendering it the dominant mode. Figure~\ref{fig:theta_est}
presents a histogram approximation of the posterior density of
$\theta_k$ conditioned on assuming that mode 2 is true,
$f_{\theta_k\mid r_k=2,\meas}$. Clearly, in this case
\begin{equation}
  \label{eq:examp_supp}
  \supp f_{\theta_k\mid r_k=2,\meas} = [0, 0.28]
\end{equation}
yielding $\hat\theta^\star = 0.28$, as marked by the red dashed line
in the figure.

In practice, greedily shooting for the maximal value of the sojourn
time, as per Eq.~\eqref{eq:theta_max_alt}, renders $\theta^\star$
highly volatile, as the maximal sojourn time might be determined by a
particle of a nearly zero weight. An alternative approach,
circumventing this volatility, is to soften the approach of
Eq.~\eqref{eq:theta_max_alt} by estimating $\theta^\star$ as the
solution of a constrained optimization problem, as follows:
	% \begin{align}
	% 	\hat{\Delta}_{k}  = \; & \min \theta \notag \\ 
	% 	& \text{s.t. }  \Pr(\theta_{k} \leq \theta,
	% 	\mod r_{k} \neq r) \geq p_{\text{Thres}}	
	% 	\label{eq:del_hat}
	% \end{align}
%
\begin{align}
  \label{eq:del_hat}
  \theta^\star & = \min\theta\notag \\
                 & \text{s.t. } \Pr(\theta_k^{r^\prime,s} \leq \theta \mid
                 r^\prime \in \Mcr, s=1,\dots,S) \geq p_{\text{Thres}}
\end{align}
where $p_{\text{Thres}}$ is a predetermined threshold. Considering
again the example from Fig.~\ref{fig:theta_est} and setting
$p_{\text{Thres}} = 0.99$, the solution of Eq.~\eqref{eq:del_hat} results
in $\hat\theta^\star_{k} = 0.26$, as marked by the black dotted line in
Fig.~\ref{fig:theta_est}.
\begin{figure}[tbh]
  \centering \includegraphics[width=3.25in]{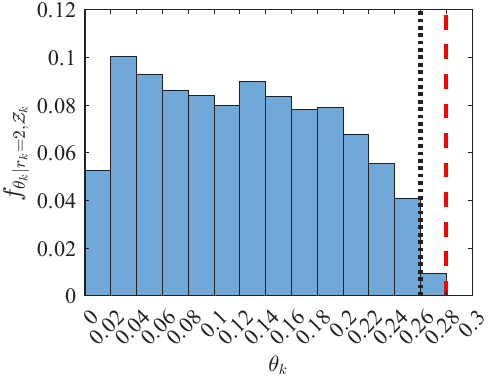}
  \caption{Histogram approximation of the posterior density of
    $\theta_k$ conditioned on assuming that mode 2 is true when mode
    1 is dominant. Red dashed line: greedy estimate of $\theta^\star$
    per Eq.~\eqref{eq:theta_max_alt}. Black dotted line: soft estimate
    of $\theta^\star$ per Eq.~\eqref{eq:del_hat}, with
    $p_\text{Thres}=0.99$.}
  \label{fig:theta_est}
\end{figure}
\item[No dominant mode exists.]  In this case, the method developed
  above cannot be used, but Eq.~\eqref{eq:theta_approx} can still
  provide a sufficiently good approximation (albeit having its cons,
  as elaborated earlier).  The approximation~\eqref{eq:theta_approx}
  can also be useful in the initialization stage of the engagement,
  when no real-time data can be used to infer the uncertainty
  interval. Yet another option is to use the latest available estimate
  of the uncertainty interval, if available, and propagate it in time
  using Eq.~\eqref{eq:delta_model}.
\end{description}

Algorithm~\ref{alg:del_est} presents a schematic description of one
cycle of the algorithm for uncertainty interval estimation. The
algorithm is initialized by setting the initialization indicator to 1.
\begin{algorithm}
	\SetAlgoNoLine
	\SetAlgoLongEnd
	\DontPrintSemicolon 
	Obtain the particles and associated weights of the current
        time step \;
	\eIf{$\exists r$ s.t. $\sum_{s=1}^{S} w_{k}^{r,s} > W_{\text{Thres}}$}{
		Solve Eq.~\eqref{eq:del_hat}
		\If{initialization = 1}{  initialization $\gets$ 0}
	}{
		\eIf{initialization = 1}{Use Eq.~\eqref{eq:theta_approx}}{Propagate the model using Eq.~\eqref{eq:delta_model}}
	}
	\caption{Uncertainty interval estimation}
	\label{alg:del_est}
\end{algorithm}
\subsection{Target-Tracking IMMPF}
\label{subsec:TTIMMPF}
The implementation of the IMMPF in our problem makes use of the models presented in Secs.~\ref{sec:Prob_Form} and~\ref{sec:SMC}. $M$ discrete modes are used for the evader acceleration command as presented in Assumption~\ref{ass6}. We recast the nonlinear EOM in Eq.~\eqref{eq:2} in discrete time as
\begin{equation}
	\textbf{x}_{k} = \textbf{f}_{k-1} (\textbf{x}_{k-1},r_{k},\textbf{w}_{k})
	\label{eq:propa}
\end{equation}
where $\textbf{x}_{k}$ is the pursuer relative state at time $t_{k}$ and $r_{k}$ is the evader mode during the time interval $( t_{k-1}, t_{k}]$.  Equation~\eqref{eq:theta} constitutes the dynamic equation of the sojourn time.  The process noise $\textbf{w}_{k}$ is treated as a tuning parameter of the filter.  The filtering step uses the measurement Eq.~\eqref{eq:Meas} to calculate the likelihood density $p(\textbf{z}_{k} \vert \textbf{x}_{k},r_{k})$.
\subsection{Particle-Smoothing Algorithm}
As one of its constitutive assumptions, the DGLCC law requires, at
each time $t$, information about $\dot{\xi}$ and $a_{E}$, delayed by
the (known) delays $\Delta_{1}(t)$ and $\Delta_{2}(t)$,
respectively. This requirement notwithstanding, common practice, as
in~\cite{shinar_nonorthodox_2002,glizer_linear_2009}, is to use the
current-time estimates, $\hat{\dot{\xi}}(t)$ and $\hat{a}_{E}$. To
examine the validity of this particular implementation, consider
Fig.~\ref{fig:xi_dot}, which presents the estimated and true values of
$\dot{\xi}$ with respect to $t_{go}$ when the evader changes its
maneuver $0.6$~s before the end of the engagement. Evidently,
apart from a particular interval after the evader changes its
maneuver, the estimate of $\dot{\xi}$ is, essentially, not delayed,
demonstrating that considering it delayed throughout the entirety of
the engagement would be erroneous.
\begin{figure}[tbh]
  \centering \includegraphics[width=3.25in]{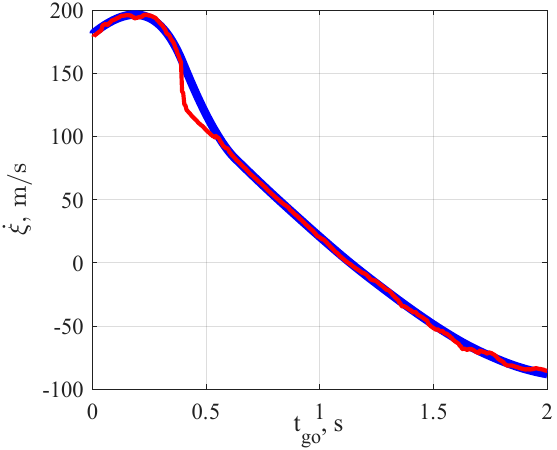}
  \caption{Estimated (red) and true (blue) $\dot{\xi}$ vs
    $t_{go}$. The evader changes its maneuver at $t_{go}= 0.6$~s.}
  \label{fig:xi_dot}
\end{figure}

Thus, common practice notwithstanding, we conclude that current-time (filtered) estimates should not be used to drive the DGLCC law. Instead, we propose using a fixed-lag smoother, which facilitates driving the guidance law with optimal delayed information as it requires.  Because the uncertainty interval is usually relatively small, we select the fixed-lag approximation smoother~\cite{kitagawa_monte_1996}, which is computationally efficient and is known to perform well for short lags.  This smoother only requires storing the previous state estimates, and the posterior smoothed states are approximated using the current-time weights; i.e., at time $t_{p}$, the posterior PDF of the states at time $t_{k}$, where $p \geq k$, is
\begin{equation}
	p\left(\textbf{x}_{k} \mid \mathcal{Z}_{p} \right)\sim
	\sum_{r=1}^{M} \sum_{s=1}^{S} w_{p}^{r,s} \delta \left(\textbf{x}-\textbf{x}_{k}^{r,s} \right)
\end{equation}
where $\mathcal{Z}_{p}$ is the set of available measurements at $t_p$.

\section{Design Considerations and Schematic Overview}
\label{sec:design}
In the previous sections, we have presented the main components of our
comprehensive approach: the extended DGLCC law (handling time-varying
delays), the specially designed estimator for this case, and the
fixed-lag smoother.  However, there are still two design aspects that
deserve consideration. The first issue is that the guidance law and,
consequently, the smoother require two time delays,
$\Delta_{1}(t_{k})$ and $\Delta_{2}(t_{k})$, whereas the estimator
intrinsically provides only the estimated lower bound of the
uncertainty interval, $\hat{\theta}^{\star}_{k}$.  Thus, a mapping
from $\hat{\theta}^{\star}_{k}$ to the two guidance-law time delays is
needed; this mapping is presented in this section. The section also
addresses the implementation of the linearized guidance law in a
nonlinear setting.
\subsection{Determination of the DGLCC Time Delays}
\label{sec:C_tun}
In Sec.~\ref{sec:GL}, we have presented a modified version of the
DGLCC law with two time-varying time delays. From
Eqs.~\eqref{eq:J1} and~\eqref{eq:J2}, the guaranteed miss distance in the
singular region is
$J^{*}_{cc} = \int^{0}_{\tau_{s}} R(\tau)\,d\tau$, where
$R(\tau) = \mu\Psi(\tau) - A(\tau)$. Since $A(\tau)$ in
Eq.~\eqref{eq:A_tau} increases monotonically with both $\Delta_{1}$
and $\Delta_{2}$, larger delays decrease $R(\tau)$, push the
singular-region boundary $\tau_{s}$ farther out, and thus increase the guaranteed
miss distance. This monotonic dependence of the game value on the
delay magnitude was established analytically for the constant-delay
case in~\cite{shinar_solution_1999} and extends qualitatively to the
time-varying case.  In summary, smaller delays are associated with
better interception performance.  However, setting the time delays too
small may degrade the smoother's performance.  The fixed-lag smoother
retrieves state estimates from the interval
$[t_{k} - \Delta_{1}(t_{k}),\; t_{k}]$.  If $\Delta_{1}(t_{k})$ is set
smaller than the true uncertainty interval, i.e., smaller than the
elapsed time since the last undetected maneuver switch, the smoother's
look-back window falls inside the postswitch transient, where the
filtered estimate carries the largest errors due to the
as-yet-undetected maneuver. In that regime, the smoothed estimate of
$\dot{\xi}$ offers marginal or no improvement over the filtered
estimate, and the resulting estimation errors increase the actual miss
distance.  Thus, a tension arises: reducing $\Delta_{1}$ improves the
game-theoretic guaranteed miss distance; however, if it is set below
the true uncertainty interval, it causes the smoother to inject
erroneous information, resulting in an increased realized miss.

Figure~\ref{fig:xi_dot_and_a_t_est} shows the behavior over time of the estimated and true $\dot{\xi}$ and $a_{E}^{n0}$ during an
engagement when the evader changes its maneuver 2~s after the engagement has begun. Whereas the estimation error of the evader acceleration increases immediately after the maneuver change, it takes approximately $0.1$~s for the estimation error of $\dot{\xi}$ to become significant.  This behavior can be explained using the
simplified model from Sec.~\ref{sec:simp_mod}:
\begin{equation}
	\tilde{a}_{E}^{n0}(t_{r}) \approx \frac{\Delta a_{E}^{c} t_{r}}{ \tau_{E}}, \qquad
	\tilde{\dot{\xi}}(t_{r})   \approx \frac{\Delta a_{E}^{c} t_{r}^{2}}{2\tau_{E}}.
	\label{eq:delay_err}
\end{equation}
For $t_{r} \ll 1$, we have $t_{r} \gg t_{r}^{2} / 2$, and the estimation
error of $a_{E}^{n0}$ is the dominant error in the first time steps.
\begin{figure}[tbh]
  \centering
  \includegraphics[width=3.25in]{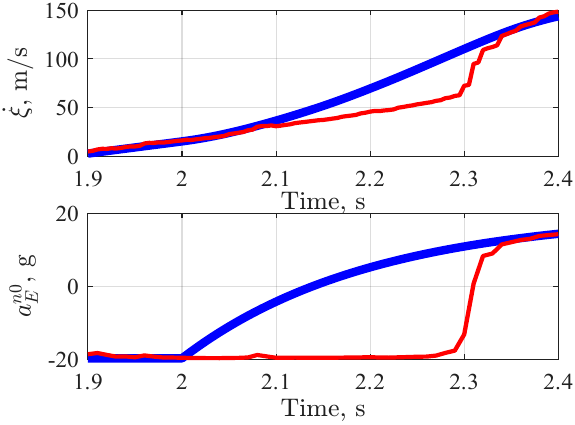}
  \caption{Estimated (red) and true (blue) $\dot{\xi}$ and
    $a_{E}^{n0}$ vs time.  The evader changes its maneuver at
    $t=2$~s.}
	\label{fig:xi_dot_and_a_t_est}
\end{figure}

Because the estimation error of $a_{E}^{n0}$ is significant even after
a few time steps, we set
\begin{equation}
	\Delta_{2}(t_{k}) =\hat{\theta}^{\star}_{k},
\end{equation}
thus ensuring that the smoother covers the entire uncertainty
interval. %\YO{YO:  notation?  Here and in the sequel.}

As for $\Delta_{1}(t_{k})$, the upper panel of Fig.~\ref{fig:xi_dot_and_a_t_est} shows that it can be set to be smaller than $\hat{\theta}^{\star}_{k}$. Thus, we propose the following model:
\begin{equation}
	\Delta_{1}(t_{k}) = C\hat{\theta}^{\star}_{k}, \quad C \in [0,1]
    \label{eq:C_constant}
\end{equation}
where $C$ is a proportionality constant that has to be determined. To
avoid the enormous computational burden associated with an extensive
Monte Carlo (MC) simulation study, we use, instead, a much more
computationally efficient nonlinear deterministic simulation based on
the nonlinear dynamics of Sec.~\ref{sec:nonlinear_dyn} and the
guidance law of Eq.~\eqref{eq:DGLCC_GL}. Moreover, we assume, for
simplicity, that the target performs only a single bang-bang maneuver,
i.e., $M=2$ during the engagement.  The pursuer is assumed to possess
perfect information, except for the information about $\dot{\xi}$ and
$a_{E}$, which is delayed according to the simplified model of
Eq.~\eqref{eq:theta_approx}. To model the effect of $C < 1$, the
information about $\dot{\xi}$ is contaminated with a bias according to
the simplified model of Eq.~\eqref{eq:delay_err} when
$t_{sw} + \Delta_1(t) < t < t_{sw} + \Delta_2(t)$,
i.e., %\YO{YO: notation}
when the smoother goes back to a time after the maneuver change but
before the estimator has detected this change. For this simulation
environment, we divide the interval $[0,1]$ into $s_{C}$ different
cases, and divide the interval $[0,t_{f}]$ into $s_{sw}$ different
cases, and then examine all possible $s_{C}\times s_{sw}$
combinations. Accordingly, we set $C$ to yield the best worst-case
performance:
\begin{equation}
	C = \arg \min_{C\in [0,1]} \max_{t_{sw} \in [0,t_{f}]} \bar{z}_{cc} (0)
	\label{eq:C_minmax}
\end{equation}
\subsection{Nonlinear Implementation of DGLCC}
The DGLCC guidance law has been developed in a linearized environment,
but the problem addressed in this paper has nonlinear dynamics, as
presented in Sec.~\ref{sec:Prob_Form}. The implementation of the
guidance relies on the ZEM, which is expressed in dimensional form as
\begin{align}
  z_{cc}(t) & = \hat{x}_{1}(t) + t_{go} \hat{x}_{2}(t - \Delta_{1}(t)) - \tau_{P}^{2} \Psi(t_{go}/\tau_{P}) x_{3}(t) - t_{go} \int_{t - \Delta_{1}(t)}^{t} x_{3}(s)ds
              \notag\\
            & + \tau_{E}^{2} \exp(-\Delta_{2}(t)/\tau_{E})[t_{go} \exp(\Delta_{1}(t)/\tau_{E}) / \tau_{E} + \exp(-t_{go}/\tau_{E}) - 1] \hat{x}_{4}(t - \Delta_{2}(t))
\end{align}
and the time-to-go is approximated using Eq.~\eqref{eq:t_go}. All the required variables are obtained using the following relations:
\begin{subequations}
	\begin{align}
		x_{1} & \dfn \xi  = \rho \sin \lambda \\
		x_{2} & \dfn \dot{\xi}  = V_{\rho} \sin \lambda - V_{\lambda} \cos \lambda \\
		x_{3} & \dfn a_{P}^{n}  = a_{P} \cos \delta_{P} \\
		x_{4} & \dfn a_{E}^{n}  = a_{E} \cos \delta_{E} 
	\end{align}
\end{subequations}
where we use the instantaneous LOS for the linearization instead of
the initial LOS. The states $x_{1}, x_{2}$, and $x_{4}$ are obtained
from the estimator and the smoother using the appropriate time
indices, and, according to Assumption~\ref{ass2}, $x_{3}$ is perfectly
known.

\subsection{Schematic Overview}
\label{sec:overview}
% Figures~\ref{fig:stand_over} and~\ref{fig:novel_over} present a
% schematic overview of the entire proposed algorithm and compare it
% with the standard implementation of guidance
% laws. 
% Figure~\ref{fig:stand_over} depicts the standard implementation of
% DGLCC, where the measurement drives the estimator, which then feeds
% the DGLCC law with the current time (filtered) state estimate. The
% DGLCC guidance law uses the current time estimate % as certainty
% % equivalent information 
% to generate an acceleration command for the
% pursuer. In this standard approach, the estimate is considered delayed
% throughout the entire engagement, and the time delays only serve as
% tuning parameters.
% %%
% \begin{figure}[tbh]
%   \centering
%   \includegraphics[width=3.25in]{block_reg.pdf}
%     \caption{Schematic overview of the standard implementation of the
%       DGLCC guidance law in a stochastic setting.}
%     \label{fig:stand_over}
% \end{figure}

Figure~\ref{fig:novel_over} presents the novel tracking and
interception strategy developed in this work. This scheme uses the
IMMPF estimator, which % can deal with nonlinear hybrid dynamics,
% non-Gaussian driving noises, and a non-Markovian switching
% mechanism. Moreover, the IMMPF
provides an approximation of the entire posterior PDF of the state,
to % . Using this approximate PDF, we
estimate the uncertainty interval in real time. This provides
real-time estimates of the two time-varying time delays that are
required by the extended DGLCC law. The time delays and the outputs of
the IMMPF are used by the fixed-lag smoother to provide
correctly timed information about delayed states. Finally, the DGLCC
guidance law uses the information from the IMMPF, the estimated time
delays, and the smoothed states to produce the pursuer's acceleration
command.
\begin{figure}[tbh]
  \centering \includegraphics[width=3.25in]{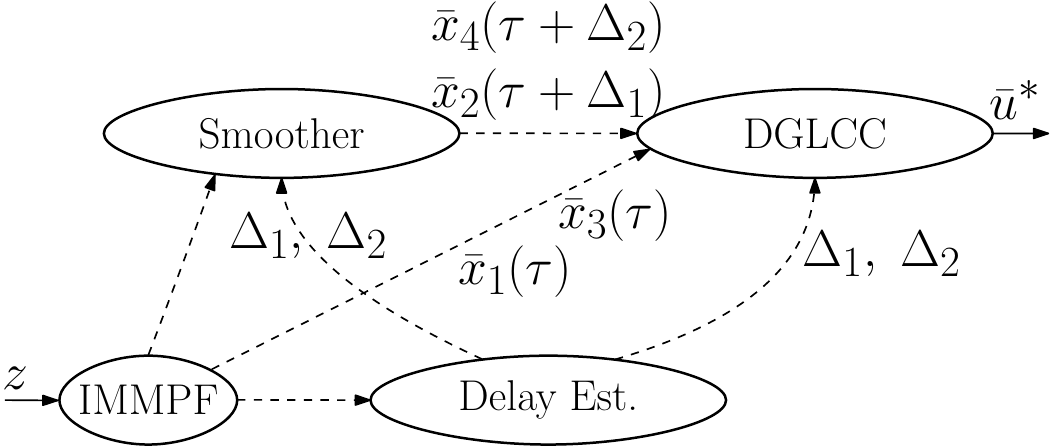}
  \caption{Schematic overview of the novel implementation of the DGLCC
    guidance law in a stochastic setting.}
  \label{fig:novel_over}
\end{figure}

\section{Simulation Study}
\label{sec:sim}
To evaluate the performance of the tracking and guidance scheme
developed in this paper, we conduct a numerical simulation study, where
we adopt the interception scenario
of~\cite{shaferman_stochastic_2016}.
Three guidance laws are compared via extensive MC simulations, namely,
DGL1, DGLC, and DGLCC with time-varying delays (TV-DGLCC).
\subsection{Simulation Scenario}
\label{subsec:param}
A ballistic missile defense scenario is considered, which
includes a single pursuing missile and a single, highly maneuverable,
evading target. For simplicity, we assume that the target performs a
maximum-acceleration bang-bang evasion maneuver with a single command
switch during the engagement, i.e., $M=2$ in this case.

The evader is initialized in the $-Y_{I}$ direction with a flight-path
angle of $\gamma_{E}(0)=-\pi/2$~rad, and the pursuer's flight path is
chosen such that the interceptor velocity vector points toward the
initial target location. The pursuer's and evader's maneuver
capabilities are $a_{P}^{\max}=45$~g and $a_{E}^{\max}=20$~g,
respectively; their first-order time constants are
$\tau_{P}=\tau_{E}=0.2$~s, and their speeds are
$V_{P}=V_{E}=2500$~m/s, rendering a nominal engagement time of $3$~s.
The measurement noise is Gaussian distributed with
$\nu \sim \mathcal{N}(0,\sigma^{2})$ with $\sigma=0.5$~mrad; the sensor's sampling rate is $f=100$~Hz.
\subsection{Filter Initialization}
The pursuer's filter is assumed to be initialized by a
radar. Figure~\ref{fig:init-radar-geometry} presents a schematic view
of the assumed planar geometry of the initializing radar, the pursuer,
and the evader. An additional subscript R denotes variables associated
with the radar. The slant range between the radar and the evader is
denoted by $\rho_{R}$, and $\lambda_{R}$ is the angle between the
radar's LOS to the evader and the $X_{I}$ axis.
\begin{figure}[tbh]
  \centering \includegraphics[width=3.25in]{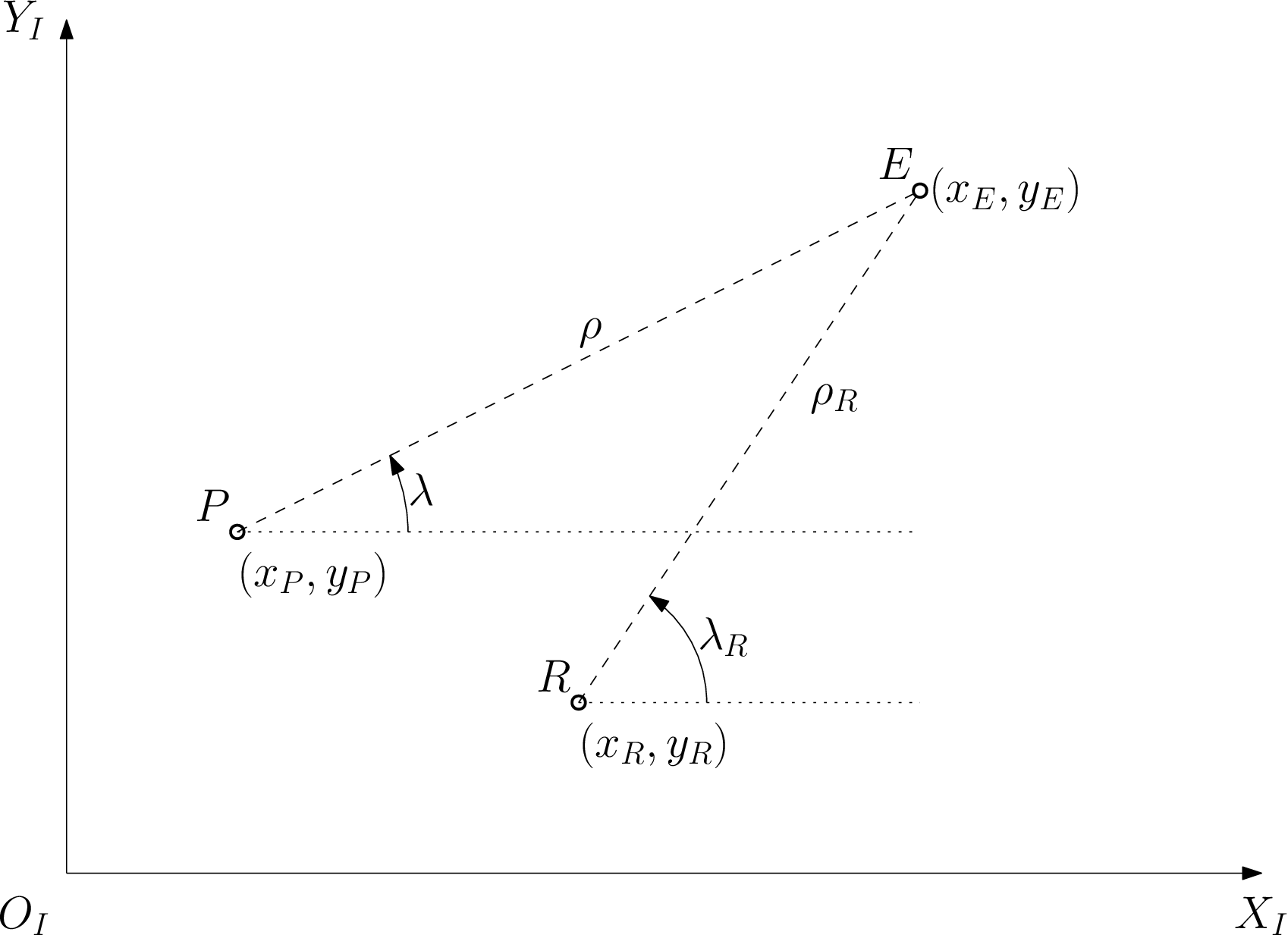}
  \caption{Initializing radar geometry.}
  \label{fig:init-radar-geometry}
\end{figure}
The radar is assumed to estimate the initial vector
\begin{equation}
  \hat{\textbf{x}}_{R} = \begin{bmatrix}
    \rho_{R} & \lambda_{R} & \gamma_{E} & a_{E} & \theta
  \end{bmatrix}^{T}.
  \label{eq:RadarSV}
\end{equation}
The radar's posterior estimate is then passed on to the pursuer at
initialization. We further assume that the radar's position
$(x_{R},y_{R})$ and the pursuer's position at initialization,
$(x_{P},y_{P})$, are known. We do not assume that the posterior estimate
$\hat{x}_{R}$ is Gaussian distributed.  The geometric relation between
the pursuer's relative state and initial radar estimate is
\begin{subequations}
  \begin{align}
    \rho &= \left\lbrace \rho_{R}^{2} + \Delta R_{R}^{2} + 2\rho_{R} [\Delta X_{R} \cos(\lambda_{R}) + \Delta Y_{R} \sin(\lambda_{R})]\right\rbrace ^{1/2} \\
    \lambda &= \arctan \left[ \frac{\Delta Y_{R} + \rho_{R} \sin(\lambda_{R})}{\Delta X_{R} + \rho_{R} \cos(\lambda_{R})} \right]
  \end{align}
  \label{eq:RadarInit}
\end{subequations}
where
\begin{equation}
	\Delta X_{R} \triangleq x_{R} - x_{P}; \quad 
	\Delta Y_{R} \triangleq y_{R} - y_{P}; \quad 
	\Delta R_{R} \triangleq ( \Delta X_{R}^2 + \Delta Y_{R}^2 )^{1/2}
\end{equation}
The radar's estimate and the relations~\eqref{eq:RadarInit} are used
to initialize the state vector of Eq.~\eqref{eq:AugStateVec}.
\subsection{Design Considerations}
The IMMPF uses $4000$ weighted particles to approximate the posterior
state vector PDF, $2000$ particles per mode.  The modes are
initialized to be equally probable. The initial estimates of the first
four states in the state vector~\eqref{eq:AugStateVec} are chosen to
satisfy
\begin{equation}
	\hat{x}_{R}\sim\mathcal{N}(\bar{x}_{R},P_{R})
	\label{eq:RadarInitPDF}
\end{equation}
where $\bar{x}_{R}$ is the true initial state, and the associated initial covariance matrix is
\begin{equation}
	P_{R}=\text{diag} \{ 50^{2},(1\pi/180)^{2},(3\pi/180)^{2},10^{2} \}.
	\label{eq:P_R}
\end{equation}
The last state, $\theta$, is initialized based on the initial evader
acceleration, $a_{E}(0) = - 20$~g.  Thus, we set the particles of the
first mode, corresponding to $a_{E}^{c} = -20$ g, to be uniformly
distributed over the interval $[2.9,3.1]$~s, as they represent the
event that a relatively long time has passed since the latest modal
switch.  Following the same reasoning, the particles of the second
mode, corresponding to $a_{E}^{c} = 20$ g, are set to be
uniformly distributed over the interval $[0,0.2]$~s, as they represent
a recent modal switch that has not been detected yet.

The TPM is chosen to be
\begin{equation}
  \Pi = \begin{bmatrix}
    0.999 & 0.001 \\
    0.001 & 0.999
  \end{bmatrix}.
\end{equation}
The transition probabilities are chosen to be small, thus reducing the
estimation bias when there are no acceleration switches.  This,
however, comes at the price of slower detection of sudden evasion
maneuvers and the eventual generation of considerable estimation time
delays.

For the estimation of the uncertainty interval, we set
$W_{\text{Thres}} = 0.9$ and $p_{\text{Thres}} = 0.99$. We then use
the following model for the time-varying delay:
\begin{equation}
	\Delta_{2} (t_{go}) = 0.01 + b\, t_{go}^{1/3}.
\end{equation}

We use the tuning algorithm presented in Sec.~\ref{sec:C_tun} to
determine the proportionality constant $C$ of
Eq.~\eqref{eq:C_constant}. Figure~\ref{fig:C_tun} presents the maximum
value of the normalized miss distance versus $C$. The plot is obtained via
a numerical simulation involving 651 cases for all possible
combinations of 21 values of $C$ over $[0, 1]$, and 31 values of the
timing of the evasion maneuver over $[0, 3]$~s. As
Fig.~\ref{fig:C_tun} shows, the solution of the minimax
problem~\eqref{eq:C_minmax} is $C = 0.75$.
\begin{figure}[tbh]
  \centering \includegraphics[width=3.25in]{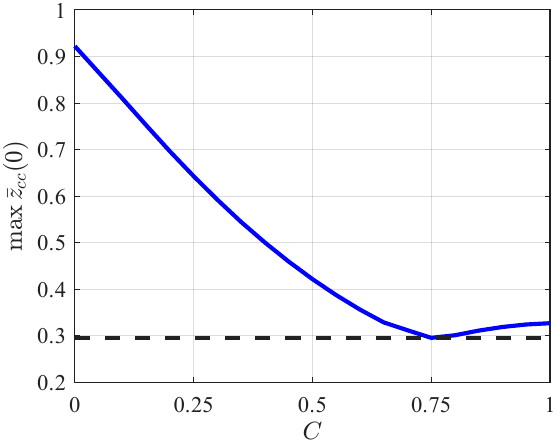}
  \caption{The maximal value of the normalized miss distance over all
    possible $31$ evasion maneuvers against all possible $21$ values
    for $C$.}
  \label{fig:C_tun}
\end{figure}

All guidance laws use a linear command inside their singular region to
prevent chattering, with $k = 0.7$. The DGLC guidance law uses a
constant time delay $\Delta t = 0.3$ s, tuned for improved performance
of the required warhead that guarantees a kill with probability 0.95.
% The estimation-aware version of DGLCC uses a horizon of a single
% time-step $h_{\max} = 0.01$ s.
%%
\subsection{Monte Carlo Analysis}
An MC simulation study is used to compare the closed-loop performance
of the interceptor using the following three guidance laws: DGL1, DGLC,
and DGLCC with time-varying delays (denoted TV-DGLCC). The evader
performs a single bang-bang evasion maneuver. To examine the effects of
the timing of this maneuver, we analyze maneuver times that cover the
entire duration of the engagement, set at 0.1~s intervals.  Two hundred MC
runs are used for each switching time. This yields a total of 6000
runs, which are performed for all the guidance laws.

Figure~\ref{fig:MD} shows the mean and standard deviation of the
pursuer's miss distance versus target maneuver time, using each of the
guidance laws.  The DGL1 law suffers the most from bang-bang evasion
maneuvers that are timed at the most challenging time window. This
stems from the fact that DGL1 does not account for the time delays
caused by the bang-bang maneuver. The DGLC guidance law also suffers
from this type of challenging maneuver. Although its performance is
better than that of DGL1, its mean miss distance is approximately five
times worse for $t_{sw} = 2.3$ s than for $t_{sw} \leq 1.5$ s, and the
standard deviation of the miss distance increases dramatically between
these cases. In the case of DGLC, this effect of the target's maneuver
timing is due to the neglect of the effect of the time delay on
$\dot{\xi}$. Moreover, in implementing the DGLC law, we have followed
the common practice, according to which the time delay is considered
constant and used as a tuning parameter. The estimate that drives the
law is obtained directly from the estimator, disregarding the design
requirement that it be delayed.

The TV-DGLCC guidance law is the least affected by the challenging
evasion maneuvers. Although its performance for $t_{sw} \leq 1.5$ is
worse than that of DGL1 and DGLC, it outperforms these two guidance
laws when challenged with better-timed evasion maneuvers.  Its mean
miss distance is approximately $1.75$ times higher for
$t_{sw} = 2.3$~s than for $t_{sw} \leq 1.5$ s, which shows that it,
too, is sensitive to the timing of the evasion maneuver. Nevertheless,
it is significantly less sensitive than both the DGL1 and DGLC laws, with
mean miss distances that are $56$ and $5$ times greater, respectively,
for $t_{sw} = 2.3$ s. The TV-DGLCC law does show somewhat increased
miss distances when the target misjudges the timing of its evasive
maneuvers, but this minor drawback is a worthwhile tradeoff for the
enhanced resilience against strategically timed maneuvers.
Furthermore, the variability of the miss distances under the TV-DGLCC
law, as indicated by their standard deviations (refer to the lower
panel of Fig.~\ref{fig:MD}), is the lowest among the laws tested. This
highlights the law's exceptional robustness concerning the target's
timing of its evasive maneuvers.
% YO: not sure I like the magenta color. Too close to red, in my
% opinion. Though I have added line types to caption.
%%
\begin{figure}[tbh]
  \centering
  \includegraphics[width=3.25in]{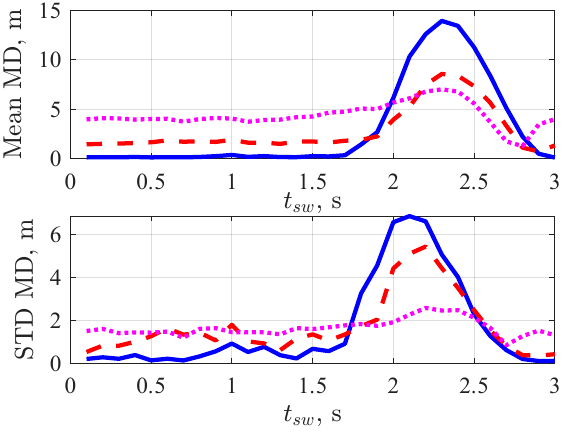}
  \caption{Miss distance mean (upper panel) and standard deviation
    (lower panel) vs evasion maneuver switching time in a nonlinear MC
    simulation. Guidance laws: DGL1 (solid blue), DGLC (dashed red),
    and TV-DGLCC (dotted magenta).}
	\label{fig:MD}
\end{figure}

Figure~\ref{fig:CDF} presents the miss distance cumulative
distribution function (CDF), computed based on all 6000 MC runs. The
figure shows that the TV-DGLCC guidance law is superior to DGL1 and
DGLC for high values of miss distances, which correspond to the most
difficult (from the standpoint of the pursuer) interception
scenarios. To guarantee a kill with a probability of 0.95, DGL1
requires a warhead with a lethality radius of 15.7~m, and DGLC
requires a warhead with a lethality radius of 10.4~m, which is a
$33.8\%$ performance improvement.  In comparison, the TV-DGLCC law
requires a warhead with a lethality radius of 8.5~m, corresponding to
$45.9\%$ and $18.3\%$ performance improvements relative to the DGL1 and
DGLC laws, respectively.
% The CDF curve of DGL1 fluctuates slightly due to the 0.1~s length
% between each evasion maneuver. Nevertheless, this artifact happens
% only for relatively small miss distances and does not affect the
% worst-case performance of DGL1.
%%
\begin{figure}[tbh]
  \centering
  \includegraphics[width=3.25in]{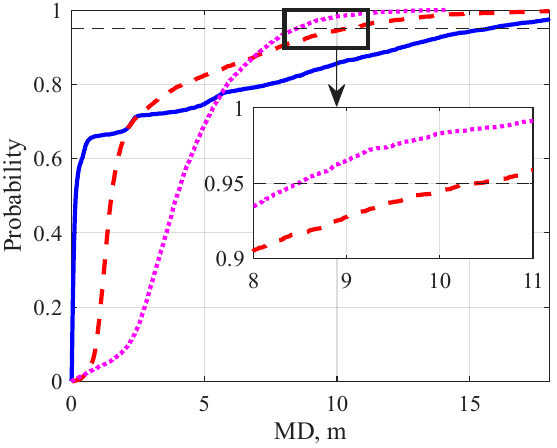}
  \caption{Miss distance CDF in a nonlinear MC simulation. Guidance
    laws: DGL1 (solid blue), DGLC (dashed red), and TV-DGLCC (dotted magenta).}
	\label{fig:CDF}
\end{figure}

\section{Conclusions}
\label{sec:concl}
This paper introduces a unified framework for tracking and
intercepting highly maneuverable evasive targets in a stochastic
setting. The main pillar of the new framework is an extended version
of the DGLCC law, which accounts for two time-varying estimation
delays, thus providing a principled way to incorporate realistic
estimator lags into the guidance framework. Two additional
contributions of the paper ensure that the guidance law is driven by
information consistent with its derivation. First, the estimation
delays, which are not constant throughout the engagement (contrary to
assumptions made in the relevant literature), are estimated online
based on a newly introduced method for estimating the uncertainty
interval following an abrupt target maneuver. Second, a particle-based
fixed-lag smoother is employed that provides the guidance law with
appropriately delayed state estimates, using the estimated delays and
all available measurements.

Together, these elements form a coherent methodology that links
estimation, delay modeling, and guidance in a self-consistent way. By
emphasizing the structural compatibility between the law and its
supporting estimation processes, the approach lays the foundation for
more robust guidance strategies in interception problems where
stochastic effects and maneuver uncertainty cannot be ignored.  Monte
Carlo simulations demonstrate that this approach improves worst-case
interception performance compared to existing laws. In particular, the
time-varying DGLCC consistently reduces the lethality radius
requirements relative to both existing DGL1 and DGLC guidance laws.

\appendix
\numberwithin{equation}{section}
%\makeatletter 
% "activate" the preparatory code, but for section-level headers only
%\newcommand{\section@cntformat}{Appendix \thesection:\ }
\renewcommand{\theequation}{\thesection\arabic{equation}}

\section{Proof of Theorem~\ref{th:1}}
\label{append3}

The proof, which follows the approach
of~\cite{glizer_linear_2009}, proceeds in three stages. First, the
delayed-control functional
$\mathcal{F}(\bar{v}_{\tau}(\cdot),\tau)$ is bounded, establishing
the admissible range of an equivalent instantaneous controller.
Second, a normalized auxiliary game with instantaneous bounded
controls is formulated. Third, standard necessary conditions for
optimality are applied to the auxiliary game, and the resulting
optimal controls are mapped back to the original delayed-control
setting.

\medskip
\noindent\textit{Stage~1: Bounding the functional $\mathcal{F}$.}
Introducing a new (instantaneous) evader controller $V(\tau)
\triangleq \mathcal{F}(\bar{v}_{\tau}(\cdot),\tau)$, the dynamics of
the center of the uncertainty set, Eq.~\eqref{eq:z_cc_dyn}, become
\begin{equation}
  d \bar{z}_{cc}/d\tau = \mu \Psi(\tau)\, \bar{u}(\tau) + V(\tau).
  \label{eq:zcc_V}
\end{equation}
The admissible range of $V(\tau)$ is determined next. Applying the
change of variable $\sigma = s - \tau$ to the
functional~\eqref{eq:F_def} yields
\begin{align}
  \mathcal{F}(\bar{v}_{\tau}(\cdot),\tau)
  & = \int_{0}^{\Delta_{1}(\tau)}
      \bigl[\exp(-\sigma/\epsilon)-1\bigr]\,
      \bar{v}_{\tau}(\sigma)\,d\sigma \notag \\
  & \quad + \int_{\Delta_{1}(\tau)}^{\Delta_{2}(\tau)}
      \Bigl[1 - \exp\!\bigl(\Delta_{1}(\tau)/\epsilon\bigr)
      \bigl[\tau(1+\gamma_{1}(\tau))/\epsilon + 1\bigr]\Bigr]\,
      \exp(-\sigma/\epsilon)\,
      \bar{v}_{\tau}(\sigma)\,d\sigma \notag \\
  & \quad - \exp(-\Delta_{2}(\tau)/\epsilon)\,
      \bigl[\tau\bigl(\exp(\Delta_{1}(\tau)/\epsilon) - 1\bigr)
      + \epsilon\Psi(\tau/\epsilon)\bigr]\,
      \bar{v}_{\tau}(\Delta_{2}(\tau))\,(1 + \gamma_{2}(\tau)),
  \label{eq:F_transformed}
\end{align}
where the integrands and the
coefficient multiplying $\bar{v}_{\tau}(\Delta_{2}(\tau))$ are all
negative. Consequently, for any $\tau \in [0,\tau_{0}]$,
\begin{subequations}
  \label{eq:F_bounds}
  \begin{align}
    \max_{\abs{\bar{v}_{\tau}(\sigma)}\leq1,\;
          \sigma\in[0,\Delta_{2}(\tau)]}
    \mathcal{F}(\bar{v}_{\tau}(\cdot),\tau)
    &= \mathcal{F}(\bar{v}_{\tau}(\cdot),\tau)
       \big\vert_{\bar{v}_{\tau}(\sigma)\equiv-1}
    = A(\tau)
    \label{eq:F_max} \\
    \min_{\abs{\bar{v}_{\tau}(\sigma)}\leq1,\;
          \sigma\in[0,\Delta_{2}(\tau)]}
    \mathcal{F}(\bar{v}_{\tau}(\cdot),\tau)
    &= \mathcal{F}(\bar{v}_{\tau}(\cdot),\tau)
       \big\vert_{\bar{v}_{\tau}(\sigma)\equiv1}
    = -A(\tau).
    \label{eq:F_min}
  \end{align}
\end{subequations}
Thus, $A(\tau) > 0$, as defined in Eq.~\eqref{eq:A_tau}, 
% \begin{align}
%   A(\tau) &\dfn \Delta_{1}(\tau) + \tau(1+\gamma_{1}(\tau))
%   + \epsilon\exp\!\bigl(-(\tau + \Delta_{2}(\tau))/\epsilon\bigr)\,
%     (1 + \gamma_{2}(\tau)) \notag \\
%   &\quad - \epsilon\exp(-\Delta_{2}(\tau)/\epsilon)\,\gamma_{2}(\tau)
%   + \exp\!\bigl((\Delta_{1}(\tau) - \Delta_{2}(\tau))/\epsilon\bigr)\,
%     \bigl[\tau(\gamma_{2}(\tau) - \gamma_{1}(\tau)) - \epsilon\bigr]
%   > 0
%   \label{eq:A_def}
% \end{align}
for all $\tau \in [0,\tau_{0}]$. It follows from Eq.~\eqref{eq:F_bounds}
that the instantaneous controller satisfies
\begin{equation}
  \abs{V(\tau)} \leq A(\tau), \quad \tau \in [0,\tau_{0}].
  \label{eq:V_bound}
\end{equation}

\medskip
\noindent\textit{Stage~2: Auxiliary game.}
Defining the normalized controller
$\bar{V}(\tau) \triangleq V(\tau)/A(\tau)$, the
bound~\eqref{eq:V_bound} becomes $|\bar{V}(\tau)| \leq 1$. An
auxiliary game is then formulated with the same cost
function~\eqref{eq:cost} and the dynamics
\begin{equation}
  d \bar{z}_{cc}/d\tau
  = \mu \Psi(\tau)\, \bar{u}(\tau) + A(\tau)\, \bar{V}(\tau)
  \label{eq:aux_dyn}
\end{equation}
subject to the control constraints
\begin{equation}
  \abs{\bar{u}(\tau)} \leq 1, \qquad
  \abs{\bar{V}(\tau)} \leq 1, \qquad
  \tau \in [0,\tau_{0}].
  \label{eq:aux_constraints}
\end{equation}

\medskip
\noindent\textit{Stage~3: Optimal controls.}
The auxiliary game [Eqs.~\eqref{eq:aux_dyn} and~\eqref{eq:aux_constraints}]
is a one-dimensional pursuit--evasion game with instantaneous bounded
controls. Applying the necessary conditions for optimality
from~\cite{shima_time-varying_2002,glizer_linear_2009} yields
\begin{subequations}
  \label{eq:aux_opt}
  \begin{align}
    \bar{u}^{*}(\tau)  &= \phantom{-}\sign \bar{z}_{cc}(0)
    \label{eq:u_opt_aux} \\
    \bar{V}^{*}(\tau) &= -\sign \bar{z}_{cc}(0).
    \label{eq:V_opt_aux}
  \end{align}
\end{subequations}
The mapping from $\bar{V}^{*}$ back to the original delayed control
$\bar{v}^{*}_{\tau}$ follows from the sign structure
of Eq.~\eqref{eq:F_transformed}. Because all integrands and coefficients
in Eq.~\eqref{eq:F_transformed} are negative, the functional
$\mathcal{F}$ reverses the sign of its input: $\bar{v}_{\tau} \equiv
c$ produces $\mathcal{F} = -c\,A(\tau)$.  Hence, the condition
$\bar{V}^{*}(\tau) = -\sign \bar{z}_{cc}(0)$ in Eq.~\eqref{eq:V_opt_aux}
is realized by
\begin{equation}
  \bar{v}^{*}_{\tau}(\sigma) = \sign \bar{z}_{cc}(0), \qquad
  \tau \in [0,\tau_{0}],\quad
  \sigma \in (0,\Delta_{2}(\tau)].
  \label{eq:v_delayed_opt}
\end{equation}
The optimal controls of the original game are, therefore,
\begin{equation}
  \bar{u}^{*}(\tau) = \bar{v}^{*}(\tau)
  = \sign \bar{z}_{cc}(0), \qquad \tau \in [0,\tau_{0}]
  \label{eq:opt_controls}
\end{equation}
and the optimal initial condition is
\begin{equation}
  \varphi^{*}(\zeta) = \bar{v}^{*}_{\tau_{0}}(\zeta - \tau_{0})
  = \sign \bar{z}_{cc}(0), \qquad
  \zeta \in (\tau_{0},\tau_{0} + \Delta_{2}(\tau_{0})].
  \label{eq:opt_ic}
\end{equation}
\qed

%%% Local Variables:
%%% mode: LaTeX
%%% TeX-master: "main"
%%% End:

\ifsupp
  \section{Perfect Information DGL1}
\label{sec:DGL1}
The underlying assumptions of this game are: 1) perfect information:
both players have complete information about the current state of the
game, and 2) there are no time delays. In this case, the game is
transformed into a one-dimensional state variable, the ZEM, which is
given by
\begin{equation}
	\bar{z}(\tau) = \bar{x}_{1}(\tau) + \tau \bar{x}_{2}(\tau) - \mu \Psi(\tau)\bar{x}_{3}(\tau) + \epsilon^{2} \Psi (\tau / \epsilon) \bar{x}_{4}(\tau)
	\label{eq:z_DGL1}
\end{equation}
where $\Psi$ is defined in~\eqref{eq:Psi}.
% where \YO{YO:  should be definition?}
% \begin{equation}
% 	\Psi(\tau) = \exp(-\tau) + \tau - 1.
% \end{equation}

The cost function is set to be the miss distance
\begin{equation}
	J = \abs{\bar{x}_{1}(0)} = \abs{\bar{z}(0)}
\end{equation}
and the optimal solution of this game is the well-known DGL1
law~\cite{gutman_optimal_1979,shinar_solution_1981}. The optimal
controls are
\begin{equation}
	\bar{u}^{*}(\tau) = \bar{v}^{*}(\tau) = \sign \bar{z}(\tau).
\end{equation}

If the pursuer possesses maneuverability and agility advantages, i.e.,
$\mu > 1$ and $\mu \epsilon > 1$, it can enforce a direct hit by
following its optimal strategy.  In this case we say that the pursuer
possesses a hit-to-kill capability.
  % \appendix
% \numberwithin{equation}{section}
% %\makeatletter 
% % "activate" the preparatory code, but for section-level headers only
% %\newcommand{\section@cntformat}{Appendix \thesection:\ }
% \renewcommand{\theequation}{\thesection\arabic{equation}}

\section{Derivation of \texorpdfstring{Eq.~\eqref{eq:z_cc}}{Eq.~(z\_cc)}}
\label{append1}

The uncertainty set is given by
\begin{equation}
    \bar{Z}(\tau) = [\bar{z}_{1}(\tau) + m(\tau),\;\bar{z}_{1}(\tau) + M(\tau)]
    \label{eq:Zset_app}
\end{equation}
where $m(\tau)$ and $M(\tau)$ are the solutions of the optimization
problems~\eqref{eq:opt_mM}. These optimization problems are subject
to the dynamic equations
\begin{subequations}
    \label{eq:dyn_app}
    \begin{align}
        d\bar{x}_{2}/d\sigma &= - \bar{x}_{4} + \mu \bar{x}_{3}
        \label{eq:norm_2} \\
        d\bar{x}_{4}/d\sigma &=  (\bar{x}_{4} - \bar{v}) / \epsilon 
        \label{eq:norm_4}
    \end{align}
\end{subequations}
with the delayed-state boundary conditions
\begin{subequations}
    \label{eq:delayed_bc}
    \begin{align}
        \bar{x}_{2}(\sigma) \big\vert_{\sigma = \tau + \Delta_{1}(\tau)} & = \bar{w}_{2}(\tau)
        \label{eq:w_2} \\
        \bar{x}_{4}(\sigma) \big\vert_{\sigma = \tau + \Delta_{2}(\tau)} & = \bar{w}_{4}(\tau).
        \label{eq:w_4}
    \end{align}
\end{subequations}

The derivation proceeds in three steps. First, the range of
admissible values for $\bar{x}_{4}$ at the intermediate time
$\tau + \Delta_{1}(\tau)$ is determined. Then, $\bar{z}_{2}(\tau)$ is
expressed as a function of these admissible values and the evader's
control on $[\tau,\,\tau + \Delta_{1}(\tau)]$. Finally, the resulting
expressions for $m(\tau)$ and $M(\tau)$ are combined to obtain the
center of the uncertainty set.

\medskip
\noindent\textit{Step~1: Admissible range of $\bar{x}_{4}(\tau + \Delta_{1}(\tau))$.}
Consider the evader's acceleration dynamics~\eqref{eq:norm_4} on the
interval $[\tau + \Delta_{1}(\tau),\,\tau + \Delta_{2}(\tau)]$, with
the boundary condition~\eqref{eq:w_4}. The extremal values of
$\bar{x}_{4}$ at $\sigma = \tau + \Delta_{1}(\tau)$ over all
admissible controls $|\bar{v}(\sigma)| \leq 1$ are
\begin{subequations}
    \label{eq:x4_minmax}
    \begin{align}
        \bar{x}_{4}^{\min}(\tau) &= \min_{\abs{\bar{v}(\sigma)}\leq1,\;\sigma\in[\tau + \Delta_{1}(\tau),\,\tau + \Delta_{2}(\tau)]} \bar{x}_{4}(\tau + \Delta_{1}(\tau)) \label{eq:x4min}\\
        \bar{x}_{4}^{\max}(\tau) &= \max_{\abs{\bar{v}(\sigma)}\leq1,\;\sigma\in[\tau + \Delta_{1}(\tau),\,\tau + \Delta_{2}(\tau)]} \bar{x}_{4}(\tau + \Delta_{1}(\tau)). \label{eq:x4max}
    \end{align}
\end{subequations}
Integrating~\eqref{eq:norm_4} with boundary condition~\eqref{eq:w_4}
from $\tau + \Delta_{2}(\tau)$ to $\tau + \Delta_{1}(\tau)$ yields
\begin{align}
    \bar{x}_{4}(\tau + \Delta_{1}(\tau)) = &\; \bar{w}_{4}(\tau)\exp\!\bigl((\Delta_{1}(\tau) - \Delta_{2}(\tau)) / \epsilon\bigr) \notag \\
    & + \frac{1}{\epsilon} \int_{\tau + \Delta_{1}(\tau)}^{\tau + \Delta_{2}(\tau)}
    \exp\!\bigl((\tau + \Delta_{1}(\tau) - s) / \epsilon\bigr)\, \bar{v}(s)\, ds.
    \label{eq:x4_integrated}
\end{align}
Because
$\exp\!\bigl((\tau + \Delta_{1}(\tau) - s)/\epsilon\bigr)/\epsilon$ is
positive for all
$s \in [\tau + \Delta_{1}(\tau),\,\tau + \Delta_{2}(\tau)]$, the
left-hand side of~\eqref{eq:x4_integrated} is minimized by
$\bar{v}(\sigma) \equiv -1$ and maximized by
$\bar{v}(\sigma) \equiv 1$. Substitution into~\eqref{eq:x4_integrated}
gives
\begin{subequations}
    \label{eq:x4_minmax_explicit}
    \begin{align}
        \bar{x}_{4}^{\min}(\tau) &= \exp\!\bigl((\Delta_{1}(\tau) - \Delta_{2}(\tau)) / \epsilon\bigr)\, (\bar{w}_{4}(\tau) + 1) - 1 \label{eq:x4min_val}\\
        \bar{x}_{4}^{\max}(\tau) &= \exp\!\bigl((\Delta_{1}(\tau) - \Delta_{2}(\tau)) / \epsilon\bigr)\, (\bar{w}_{4}(\tau) - 1) + 1. \label{eq:x4max_val}
    \end{align}
\end{subequations}

\medskip
\noindent\textit{Step~2: Expression for $\bar{z}_{2}(\tau)$.}
Let $\kappa$ denote the value of $\bar{x}_{4}$ at the intermediate
boundary,
\begin{equation}
    \kappa \triangleq \bar{x}_{4}(\sigma) \big\vert_{\sigma = \tau + \Delta_{1}(\tau)},
    \qquad
    \kappa \in [\bar{x}_{4}^{\min}(\tau),\;\bar{x}_{4}^{\max}(\tau)].
    \label{eq:kappa}
\end{equation}
With this notation, the minimization problem for $m(\tau)$
in~\eqref{eq:opt_mM} reduces to
\begin{equation}
    m(\tau) = \min_{\kappa \in [\bar{x}_{4}^{\min}(\tau) ,\, \bar{x}_{4}^{\max}(\tau)]} \;\min_{\abs{\bar{v}(\sigma)}\leq1,\;\sigma\in[\tau,\,\tau + \Delta_{1}(\tau)]} \bar{z}_{2}(\tau).
    \label{eq:m_reduced}
\end{equation}

Integrating~\eqref{eq:norm_4} with initial
condition~\eqref{eq:kappa} from $\tau + \Delta_{1}(\tau)$ backward to
an arbitrary $\sigma \in [\tau,\,\tau + \Delta_{1}(\tau))$ gives
\begin{equation}
    \bar{x}_{4}(\sigma) =  \kappa \exp\!\bigl((\sigma - \tau - \Delta_{1}(\tau)) / \epsilon\bigr) + \frac{1}{\epsilon} \int_{\sigma}^{\tau + \Delta_{1}(\tau)}
    \exp\!\bigl((\sigma - s) / \epsilon\bigr)\, \bar{v}(s)\, ds.
    \label{eq:x_4_append}
\end{equation}
Next, substituting~\eqref{eq:x_4_append} into the velocity
dynamics~\eqref{eq:norm_2} and integrating from
$\tau + \Delta_{1}(\tau)$ to
$\sigma \in [\tau,\,\tau + \Delta_{1}(\tau))$, with boundary
condition~\eqref{eq:w_2}, yields
\begin{align}
    \bar{x}_{2}(\sigma) =&\; \bar{w}_{2}(\tau) - \mu \int_{\tau}^{\tau + \Delta_{1}(\tau)} \bar{x}_{3}(s)\,ds + \kappa \epsilon \bigl(1 - \exp\!\bigl((\sigma - \tau - \Delta_{1}(\tau)) / \epsilon\bigr)\bigr) \notag\\
    & + \int_{\sigma}^{\tau + \Delta_{1}(\tau)} \bigl(1 - \exp\!\bigl((-s + \sigma)/\epsilon\bigr)\bigr)\, \bar{v}(s)\,ds.
    \label{eq:x2_intermediate}
\end{align}

Substituting~\eqref{eq:x2_intermediate}
and~\eqref{eq:x_4_append} into the definition of
$\bar{z}_{2}(\tau)$, Eq.~\eqref{eq:z_2}, and evaluating at
$\sigma = \tau$, produces
\begin{equation}
    \bar{z}_{2}(\tau) = a_{1}(\tau) + \kappa\, a_{2}(\tau) - \int_{\tau}^{\tau + \Delta_{1}(\tau)} a_{3}(\tau,s)\,\bar{v}(s)\,ds
    \label{eq:z2_linear}
\end{equation}
with the coefficient functions
\begin{subequations}
    \label{eq:a_coeffs}
    \begin{align}
        a_{1}(\tau) &= \tau\left( \bar{w}_{2}(\tau) - \mu \int_{\tau}^{\tau + \Delta_{1}(\tau)} \bar{x}_{3}(s)\,ds \right) \label{eq:a1}\\
        a_{2}(\tau) &= \epsilon \tau \bigl(1 - \exp(-\Delta_{1}(\tau)/\epsilon)\bigr) + \epsilon^{2}\Psi(\tau/\epsilon)\exp(-\Delta_{1}(\tau)/\epsilon) \label{eq:a2}\\
        a_{3}(\tau,s) &= \tau\bigl(1 - \exp((\tau-s)/\epsilon)\bigr) + \epsilon \Psi(\tau/\epsilon)\exp((\tau-s)/\epsilon). \label{eq:a3}
    \end{align}
\end{subequations}
It can be verified that $a_{2}(\tau) > 0$ for all $\tau > 0$, and
$a_{3}(\tau,s) > 0$ for all $\tau > 0$ and
$s \in [\tau,\,\tau + \Delta_{1}(\tau)]$.

\medskip
\noindent\textit{Step~3: Center of the uncertainty set.}
Equation~\eqref{eq:z2_linear} is linear in both $\kappa$ and
$\bar{v}(\cdot)$. Because $a_{2}(\tau) > 0$ and $a_{3}(\tau,s) > 0$,
the minimum of $\bar{z}_{2}(\tau)$ is attained at
$\kappa = \bar{x}_{4}^{\min}(\tau)$ with
$\bar{v}(\sigma) \equiv -1$ for
$\sigma \in [\tau,\,\tau + \Delta_{1}(\tau)]$, and the maximum is
attained at $\kappa = \bar{x}_{4}^{\max}(\tau)$ with
$\bar{v}(\sigma) \equiv 1$. Accordingly,
\begin{align}
    m(\tau) &= a_{1}(\tau) + \bar{x}_{4}^{\min}(\tau)\, a_{2}(\tau) - \int_{\tau}^{\tau + \Delta_{1}(\tau)} a_{3}(\tau,s)\,ds \label{eq:m_val}\\
    M(\tau) &= a_{1}(\tau) + \bar{x}_{4}^{\max}(\tau)\, a_{2}(\tau) + \int_{\tau}^{\tau + \Delta_{1}(\tau)} a_{3}(\tau,s)\,ds. \label{eq:M_val}
\end{align}

The center of the uncertainty set~\eqref{eq:Zset_app} is, therefore,
\begin{equation}
    \bar{z}_{cc}(\tau) = \bar{z}_{1}(\tau) + \frac{m(\tau) + M(\tau)}{2} = \bar{z}_{1}(\tau) + a_{1}(\tau) + \frac{1}{2}\, a_{2}(\tau) \bigl(\bar{x}_{4}^{\min}(\tau) + \bar{x}_{4}^{\max}(\tau)\bigr).
    \label{eq:zcc_intermediate}
\end{equation}
Substituting the explicit expressions for
$\bar{x}_{4}^{\min}(\tau)$ and $\bar{x}_{4}^{\max}(\tau)$
from~\eqref{eq:x4_minmax_explicit}, together with the definitions of
$a_{1}(\tau)$ and $a_{2}(\tau)$
from~\eqref{eq:a1}--\eqref{eq:a2}, into~\eqref{eq:zcc_intermediate}
yields Eq.~\eqref{eq:z_cc} after straightforward algebraic
simplification.

\section{Derivation of
  \texorpdfstring{Eq.~\eqref{eq:z_cc_star}}{Eq.~(z\_cc\_star)}}
\label{append2}

The evolution equation for $\bar{z}_{cc}(\tau)$ is obtained by
differentiating Eq.~\eqref{eq:z_cc} with respect to $\tau$, yielding
\begin{align}
    \frac{d \bar{z}_{cc}(\tau)}{ d \tau} &  = \mu \Psi(\tau) \bar{u}(\tau) - \bar{x}_{2}(\tau) + \bar{x}_{2}(\tau + \Delta_{1}(\tau)) \notag \\
    - &\mu \int_{\tau}^{\tau + \Delta_{1}(\tau)} \bar{x}_{3}(s)\,ds
        - \tau \bar{x}_{4}(\tau + \Delta_{1}(\tau))  (1+\gamma_{1}(\tau)) \notag \\
    - &\exp(-\Delta_{2}(\tau)/\epsilon)\,[\tau \exp(\Delta_{1}(\tau)/\epsilon) + \epsilon\exp(-\tau/\epsilon) - \epsilon]\, \bar{v}(\tau + \Delta_{2}(\tau))\,(1 + \gamma_{2}(\tau)) \notag \\
    + &\exp(-\Delta_{2}(\tau)/\epsilon) \bigl[
    \exp(\Delta_{1}(\tau)/\epsilon)(\gamma_{1}(\tau)\tau + \epsilon + \tau) - \epsilon \bigr]\,
    \bar{x}_{4}(\tau + \Delta_{2}(\tau)).
    \label{eq:dzcc_raw}
\end{align}

Integrating the normalized dynamics~\eqref{eq:norm_2} from $\tau$ to
$\tau + \Delta_{1}(\tau)$ shows that
\begin{equation}
    \bar{x}_{2}(\tau + \Delta_{1}(\tau)) - \bar{x}_{2}(\tau) - \mu \int_{\tau}^{\tau + \Delta_{1}(\tau)} \bar{x}_{3}(s)\,ds = - \int_{\tau}^{\tau + \Delta_{1}(\tau)} \bar{x}_{4}(s)\,ds.
    \label{eq:x2_identity}
\end{equation}
Substituting~\eqref{eq:x2_identity} into~\eqref{eq:dzcc_raw}
simplifies the first three terms, giving
\begin{align}
    \frac{d \bar{z}_{cc}(\tau)}{ d \tau} & = \mu \Psi(\tau) \bar{u}(\tau) - \int_{\tau}^{\tau + \Delta_{1}(\tau)} \bar{x}_{4}(s)\,ds
    - \tau \bar{x}_{4}(\tau + \Delta_{1}(\tau))  (1+\gamma_{1}(\tau))  \notag\\
    - &\exp(-\Delta_{2}(\tau)/\epsilon)\,[\tau \exp(\Delta_{1}(\tau)/\epsilon) + \epsilon\exp(-\tau/\epsilon) - \epsilon]\, \bar{v}(\tau + \Delta_{2}(\tau))\,(1 + \gamma_{2}(\tau))  \notag\\
    + &\exp(-\Delta_{2}(\tau)/\epsilon) \bigl[
        \exp(\Delta_{1}(\tau)/\epsilon)(\gamma_{1}(\tau)\tau + \epsilon + \tau) - \epsilon \bigr]\,
        \bar{x}_{4}(\tau + \Delta_{2}(\tau)).
    \label{eq:dzcc_step1}
\end{align}

Next, the integral of $\bar{x}_{4}$
in~\eqref{eq:dzcc_step1} is expanded using the
identity~\cite{glizer_linear_2009}
\begin{align}
    \int_{\tau}^{\tau + \Delta_{1}(\tau)} \bar{x}_{4}(s)\,ds & = \epsilon\, \bar{x}_{4}(\tau + \Delta_{1}(\tau))\, \bigl(1 - \exp(-\Delta_{1}(\tau)/\epsilon)\bigr)  \notag\\ 
    & \quad + \int_{\tau}^{\tau+\Delta_{1}(\tau)}\bigl(1-\exp((\tau-s)/\epsilon)\bigr)\,\bar{v}(s)\,ds
    \label{eq:x4_integral_identity}
\end{align}
which, upon substitution into~\eqref{eq:dzcc_step1}, yields
\begin{align}
    \frac{d \bar{z}_{cc}(\tau)}{ d \tau} & = \mu \Psi(\tau) \bar{u}(\tau) - \int_{\tau}^{\tau+\Delta_{1}(\tau)}\bigl(1-\exp((\tau-s)/\epsilon)\bigr)\,\bar{v}(s)\,ds \notag\\
    - &\bar{x}_{4}(\tau + \Delta_{1}(\tau))\bigl[ \epsilon\bigl(1 - \exp(-\Delta_{1}(\tau)/\epsilon)\bigr) + \tau(1+\gamma_{1}(\tau)) \bigr] \notag\\
    + &\exp(-\Delta_{2}(\tau)/\epsilon) \bigl[
    \exp(\Delta_{1}(\tau)/\epsilon)(\gamma_{1}(\tau)\tau + \epsilon + \tau) - \epsilon \bigr]\,
    \bar{x}_{4}(\tau + \Delta_{2}(\tau))  \notag\\
    - &\exp(-\Delta_{2}(\tau)/\epsilon)\,[\tau \exp(\Delta_{1}(\tau)/\epsilon) + \epsilon\exp(-\tau/\epsilon) - \epsilon]\, \bar{v}(\tau + \Delta_{2}(\tau))\,(1 + \gamma_{2}(\tau)).
    \label{eq:dzcc_step2}
\end{align}

Finally, the term $\bar{x}_{4}(\tau + \Delta_{1}(\tau))$
in~\eqref{eq:dzcc_step2} is eliminated by applying the
identity
\begin{align}
    \bar{x}_{4}(\tau + \Delta_{1}(\tau)) & = \bar{x}_{4}(\tau + \Delta_{2}(\tau))\, \exp\!\left(\frac{\Delta_{1}(\tau) - \Delta_{2}(\tau)}{\epsilon}\right)  \notag\\
    & \quad + \frac{1}{\epsilon} \int_{\tau + \Delta_{1}(\tau)}^{\tau + \Delta_{2}(\tau)} \exp\!\left(\frac{\tau + \Delta_{1}(\tau) - s}{\epsilon}\right) \bar{v}(s)\,ds
    \label{eq:x4_shift_identity}
\end{align}
which is obtained by integrating~\eqref{eq:norm_4} from
$\tau + \Delta_{2}(\tau)$ to $\tau + \Delta_{1}(\tau)$.
Substituting~\eqref{eq:x4_shift_identity}
into~\eqref{eq:dzcc_step2} and collecting terms produces
Eq.~\eqref{eq:z_cc_star}.

\fi

\section*{Acknowledgments}
\label{sec:aknow}
The authors would like to thank Valery Glizer and
Vladimir Turetsky, who authored the original version of the DGLCC
guidance law, for their valuable insights during the discussion of
this paper.

\bibliography{Bib1}

\end{document}